\documentclass[journal,twoside,web]{ieeecolor}
\input{Sources/header.tex}
\newcommand{\f}{f}          
\newcommand{\Rp}{\R_{\geq 0}}      
\newcommand{\Np}{\N_{>0}}      
\renewcommand{\t}{i}      
\newcommand{\ztm}{\mathbb{Z}_{[0,T-1]}{}} 
\newcommand{\zt}{\mathbb{Z}_{[0,T]}{}} 

\newcommand{\norm}[1]{\|{#1}\|}

\newcommand{\R}{\mathbb{R}}
\newcommand{\Rx}{\mathbb{R}^{n_x}}
\newcommand{\Ru}{\mathbb{R}^{n_u}}
\newcommand{\Z}{\mathbb{Z}}
\newcommand{\N}{\mathbb{N}}
\newcommand{\PDset}[1]{\mathbb{S}_{>0}^{#1}}
\newcommand{\PSDset}[1]{\mathbb{S}_{+}^{#1}}

\newcommand{\Rxtm}{\R^{(T-1)n_x}}
\newcommand{\Rutm}{\R^{(T-1)n_u}}
\newcommand{\Rxt}{\R^{Tn_x}}

\newcommand{\x}{\bar{x}}

\renewcommand{\c}{c_1}
\newcommand{\kinf}{\mathcal{K}_\infty}

\makeatletter
\newtheoremstyle{problem}
{}              
{}              
{}        		
{}              
{\bfseries}     
{}              
{ }             
{\thmname{#1}\thmnote{ {\the\thm@notefont{\normalfont #3}}}.}
\makeatother

\floatstyle{plain}
\newfloat{opt}{htp}{los}[opt]

\crefname{assumption}{Assumption}{Assumptions}
\crefname{lemma}{Lemma}{Lemmas}
\crefname{theorem}{Theorem}{Theorems}
\crefname{section}{Section}{Sections}
\crefname{proposition}{Proposition}{Propositions}
\crefname{equation}{}{}
\crefrangeformat{equation}{(#3#1#4)-(#5#2#6)}
\crefrangeformat{assumption}{Assumptions #3#1#4-#5#2#6}
\crefname{algorithm}{Algorithm}{Algorithms}
\crefname{problem}{Problem}{Problems}
\crefname{subsection}{Subsection}{Subsections}
\crefname{figure}{Figure}{Figures}

\newtheorem{theorem}{Theorem}
\newtheorem{proposition}{Proposition}
\newtheorem{lemma}{Lemma}
\newtheorem{remark}{Remark}
\newtheorem{assumption}{Assumption}
\theoremstyle{problem}
\newtheorem{problem}{Problem}

\newcommand{\Pabs}{\hyperlink{prob:Pabs}{$\mathcal{P}$}\xspace}
\newcommand{\Pabsw}{\hyperlink{prob:Pabsw}{$\mathcal{P}_\delta$}\xspace}
\newcommand{\Pilc}{\hyperlink{prob:Pilc}{$\mathcal{P}_\text{ILC}$}\xspace}
\newcommand{\Pmpc}{\hyperlink{prob:Pmpc}{$\mathcal{P}_\text{MPC}$}\xspace}

\DeclareMathOperator*{\argmin}{arg\,min}

\newcommand\scalemath[2]{\scalebox{#1}{\mbox{\ensuremath{\displaystyle #2}}}}

\begin{document}

\title{Iterative Learning Predictive Control for Constrained Uncertain Systems}

\author{Riccardo Zuliani, Efe C. Balta, Alisa Rupenyan, John Lygeros
\thanks{Research supported by the Swiss National Science Foundation under NCCR Automation (grant agreement 51NF40\_180545). R. Zuliani and J. Lygeros are with the Automatic Control Laboratory (IfA), ETH Z\"urich, 8092 Z\"urich, Switzerland \texttt{\small$\{$rzuliani,lygeros$\}$@ethz.ch}. E. C. Balta is with Inspire AG, 8005 Z\"urich, Switzerland \& with IfA \texttt{\small efe.balta@inspire.ch}. A. Rupenyan is with ZHAW Centre for Artificial Intelligence, Z\"urich University of Applied Sciences, Switzerland \texttt{\small alisa.rupenyan@zhaw.ch}.}}


\maketitle
\begin{abstract}
Iterative learning control (ILC) improves the performance of a repetitive system by learning from previous trials. ILC can be combined with Model Predictive Control (MPC) to mitigate non-repetitive disturbances, thus improving overall system performance. However, existing approaches either assume perfect model knowledge or fail to actively learn system uncertainties, leading to conservativeness. To address these limitations we propose a binary mixed-integer ILC scheme, combined with a convex MPC scheme, that ensures robust constraint satisfaction, non-increasing nominal cost, and convergence to optimal performance. Our scheme is designed for uncertain nonlinear systems subject to both bounded additive stochastic noise and additive uncertain components. We showcase the benefits of our scheme in simulation.
\end{abstract}

\section{Introduction}\label{section:INTRO}
Iterative learning control (ILC) is an established control technique for repetitive systems that utilizes information from previous iterations to improve control performance \cite{bristow2006survey}. Since the entire time sequence from previous iterations is available at the time of the input update, ILC can achieve perfect tracking even with little a-priori knowledge of the system \cite{ahn2007iterative}.

Recently, research has focused on \emph{optimization based ILC}, which explicitly considers input and state constraints by formulating the update as an optimization problem \cite{chu2010iterative,mishra2010optimization}. Generally, however, optimization-based schemes lack guarantees of constraint satisfaction in the presence of noise and/or model uncertainty. One notable exception is \cite{liao2022robustness}, where the authors popose a forward-backward splitting algorithm and use constraint tightening to guarantee robust constraint satisfaction in the presence of uncertainty. Similar settings have been extended to online ILC methods \cite{balta2022regret} and to iteration varying systems \cite{balta2024iterative}. Despite being robust to uncertainties, these approaches do not counteract nonrepetitive disturbances nor learn the unknown components, which introduces conservatism.

Active learning can improve the performance of optimization-based schemes in settings with model uncertainty and hard input and state constraints by reducing the constraint tightening as the uncertainty diminishes. Active learning is common in Model Predictive Control (MPC) \cite{soloperto2020augmenting,soloperto2022guaranteed}. In this work, we seek to combine ILC and MPC to obtain a scheme that learns an unknown component of the system dynamics and reduces conservativeness as much as possible.

Active learning has been considered in the field of adaptive ILC. For example, \cite{chi2008adaptive} combines ILC with recursive least squares estimation to learn an unknown parameter describing a nonlinear system. \cite{chi2015data} can operate in a fully model-free setting by using data to learn the partial derivatives of the functions mapping inputs to outputs. Both these methods, however, don't consider input or state constraints.

Our paper is not the first example of combination of ILC and MPC, this is indeed a well-established area of research (see \cite{lu2019110th} for a comprehensive review). For example, the scheme in \cite{lee_model_1999} performs a norm-optimal ILC update at each time-step in a shrinking horizon fashion, where measurements from the ongoing iteration improve robustness towards non-repetitive disturbances, and information from past iterations is used to compensate repetitive disturbances. The convergence properties of this scheme presented in \cite{lee2000convergence}, however, require no model uncertainty and no iteration-varying disturbances. The method in \cite{lee_model_1999} has recently been applied to the problem of controlling the temperature of an additive manufacturing machine in \cite{zuliani2022batch}. In \cite{wang2010model}, ILC is combined with MPC to control an artificial pancreatic cell, without theoretical guarantees. The authors of \cite{rosolia2017learning} propose a learning, reference-free MPC that uses hystorical data to iteratively construct a set of states from which the control task can be completed using a known safe control policy. At the same time, the scheme approximately builds a value function that maps each visited state to the closed-loop cost for completing the task. The approach converges to an optimal trajectory while retaining recursive feasibility and constantly improving the closed-loop performance. More recently, \cite{smith2023iterative} proposed an ILC scheme that iteratively tightens the constraints of an MPC scheme to incentivise safety, but without providing guarantees.

All aformentioned methods, however, assume perfect model knowledge or lack active learning, resulting in unnecessary conservativeness. On the other hand, MPC has been coupled with adaptive control techniques to obtain schemes that simultaneously performs closed-loop identification and regulation of constrained linear systems subject to parametric and additive uncertainty \cite{weiss2014robust,aswani2013provably}. These schemes are often implemented with set-membership identification techniques \cite{blanchini2008set}, providing both a nominal model of the plant as well as a quantification of the residual uncertainty, ususally in the form of an uncertainty set. Generally, constraints are guaranteed to be satisfied for all possible uncertainty realizations inside the uncertainty set, and the uncertainty set is reduced as new observations of states and inputs become available \cite{lorenzen2017adaptive,bujarbaruah2019adaptive}. The uncertainty set should contains all admissible models that are consistent with the measurements, and can then be exploited for robust control design. 

\subsubsection*{Contribution}
In this paper we adapt the idea of set-membership identification to the iterative setting by establishing an ILC scheme that actively learns an uncertain, state-dependent component while ensuring robust constraint satisfaction of state and input constraints. The uncertainty set is used to tighten the constraints, and as the set shrinks in size, both the tightening and the conservativeness are reduced. We consider uncertain nonlinear systems subject to both state-dependent additive disturbances and bounded stochastic additive noise.

The ILC scheme is combined with an MPC scheme operating in receding horizon in each iteration. The MPC optimizes online a nominal state prediction as well as a sequence of tubes which are guaranteed to contain the true state of the system. The size of the tubes grows as the state prediction deviates from past observed trajectories. Asymptotically, the state trajectory is shown to converge to a neighborhood of the desired reference.


At the beginning of every iteration, the ILC update is formulated as a mixed binary integer program where the number of binary variables is equal to the number of time-steps considered in the optimization. The MPC problem is convex and can be solved efficiently online at each time-step with shrinking prediction horizon.

The scheme is recursively feasible and with a non-increasing nominal iteration cost. In addition, we show that the scheme converges to the same solution as a scheme with perfect model knowledge. Theoretical results are also validated in simulation.

\subsubsection*{Outline}
The remainder of this paper is structured as follows: in \cref{section:PF} we introduce the problem, in \cref{section:PREF} we propose a tractable reformulation. \cref{section:BLMPC} describes the control scheme, whose properties are later analyzed in \cref{section:AN}. In \cref{section:SIM} we present the simulation results.
\subsubsection*{Notation}
We denote with $\Z,\N,\Np$ the set of integers, non-negative integers, and positive integers, respectively.
We denote with $\Z_{[i,j]}$ the set of integers between $i$ and $j$, with $i,j\in\Z$, $i \leq j$, i.e., $\Z_{[i,j]}=\{ n\in \Z : i \leq n \leq j \}$. 
We denote with $\R,\R^n,\R^{n \times m}$ the set of real numbers, the set of real-valued $n$-dimensional vectors, and the set of $n \times m$ real-valued matrices, respectively, with $n,m\in\Np$.
$\|\cdot\|$ denotes the standard $p$-norm, and $\mathcal{B}(r)=\{ x: \|x\| \leq r \}$.
$\mathbb{S}_{>0}^n$ denotes the cone of symmetric positive semidefinite $n \times n$ real-valued matrices.
$\mathcal{K}_\infty$ denotes the class of strictly increasing functions $\alpha:[0,\infty)\to[0,\infty)$ that are zero at zero, and for which $\lim_{t \to \infty} \alpha(t)=\infty$.

\section{Problem formulation}\label{section:PF}
Consider a discrete time nonlinear time-varying system
\begin{align}
x_k(t+1) = f(x_k(t),t) + A(t) x_k(t) + B(t) u_k(t) + w_k(t),%
\label{eq:PF_system}
\end{align}
where $t\in\ztm$ denotes the time, for a finite horizon $T\in\Np$, and $k\in\N$ denotes the iteration. The variables $x_k(t)\in\Rx$, $u_k(t)\in\Ru$, and $w_k(t)\in\Rx$ denote, respectively, the state, input, and noise. We denote $d_k(t)=f(x_k(t),t)$ and refer to $d_k(t)\in\Rx$ as the disturbance.

\emph{Iteration} refers to the execution of \cref{eq:PF_system} for all $t\in\zt$ for a given $x_k(0)$ and $u_k:=(u_k(0),\dots,u_k(T-1))\in\R^{(T-1)n_u}$. Upon reaching $t=T$, $k$ is increased to $k+1$, $t$ is reset to $t=0$, and the initial state $x_{k+1}(0)$ is reset as follows.
\begin{assumption}
\label{ass:PF_initial_condition}%
For every $k \geq 0$, $x_k(0)\in \{ \bar{x} \} \oplus \mathcal{B}(r_0)$ for some known $\bar{x}\in\Rx$ and $r_0\in\Rp$.
\end{assumption}
In this paper we allow the $p$-norm $\|\cdot\|$ to have $p=2$ or $p=\infty$. Unless otherwise specified, we assume $p=\infty$ and use $\|\cdot\|:=\|\cdot\|_\infty$ for simplicity; all results of the paper can be easily adapted to the case $p=2$. The system must satisfy $(x_k(t),u_k(t)) \in \mathscr{X}$ for $t\in\ztm$ and $x_k(T)\in \mathscr{X}_T$, where
\begin{align*}
\mathscr{X} &:= \{ (x,u)\in\R^{n_x+n_u} : H_x x + H_u u \leq h \}, \\
\mathscr{X}_T &:= \{ x\in\R^{n_x} : H_{x,T} x \leq h_T \},
\end{align*}
and $H_x,H_u,h,H_{x,T},h_T$ are matrices of appropriate dimensions. We define for simplicity
\begin{align*}
\mathscr{Z} &:= \{ (x,u) : (x(t),u(t))\in \mathscr{X}, t\in\ztm, x(T)\in \mathscr{X}_T \},\\
\mathscr{X}_x &:= \{ x : \exists\, u\in\R^{n_u}: (x,u)\in \mathscr{X} \}.
\end{align*}
Two additive terms act on \cref{eq:PF_system} at each time-step. The function $f$ is assumed to be an unknown function of $t$ and $x_k(t)$, as formalized next.
\begin{assumption}
\label{ass:PF_disturbance}
For all $k$ and $t$, $d_k(t)=\f(x_k(t),t)$, where $\f(\cdot,t)$ is $L(t)$-Lipschitz in $p$-norm as a distance metric, i.e.
\begin{align*}
\| f(x,t)-f(y,t) \| \leq L(t) \|x-y\|, ~ \forall t \in \Np, ~ x,y\in \mathscr{X}_x.
\end{align*}
Additionally, $L(t) \leq m(t)$ for some known $m(t)<\infty$ for all $t$.
\end{assumption}
%
To ensure that the true system dynamics can be learned using a set-membership identification algorithm, we assume that the noise $w_k(t)$ satisfies the following.
\begin{assumption}
\label{ass:PF_noise}
For all $k$ and $t$, $w_k(t)\in\mathcal{W}=\mathcal{B}(\bar{w})$ for some $\bar{w}\in\Rp$. Moreover, 
$\operatorname{Pr}\left\{ w_k(t)\in U \right\} > 0$ for all $U \subseteq \mathcal{W}$ with $\operatorname{int}U \neq \emptyset$.
\end{assumption}
%
%
This assumption is common in the set-membership identification literature (see for example \cite{lu2021robust}), and it could be removed at the cost of a non-vanishing estimation error \cite[Section 5.3]{lu2021robust}. This is, however, beyond the scope of this paper.

We address the problem of steering the state \cref{eq:PF_system} toward a pre-defined time-varying (but iteration-invariant) reference $r(t)\in\Rx$ by adjusting the control input $u_k(t)$. To this end, we assume that the state $x_k(t)$ is measured at each time-step. Specifically, for a given $r=(r(0),\dots,r(T))$ and $Q\in\PDset{n_x}$, the goal is to solve problem \Pabs.
\begin{problem}[$\mathcal{P}$]
    \label{prob:Pabs}
    \begin{align*}
        \underset{u,x}{\text{min.}} & \quad \norm{x-r}^2_Q \\
        \text{s.t.} & \quad x(t+1) = A(t) x(t) + B(t) u(t) + \f(x(t),t) + w(t), \\
        & \quad (x,u)\in \mathscr{Z}, \\
        & \quad x(0) \text{~given}.%
    \end{align*}
\end{problem}
%

\section{Problem reformulation}\label{section:PREF}
Problem \Pabs cannot be solved due to the presence of the unknown terms $f$ and $w$. This issue can be addressed by introducing an algorithm that learns $f$ and by robustifying the problem against all possible values of $w$.


\subsection{Approximating the disturbance}
\label{subsection:PREF_disturbance}
%

To obtain the nominal dynamics, we set $w_k(t)=0$ for all $t$ in \cref{eq:PF_system} and assume $f(x_k^\text{ref}(t),t)=d_k^\text{ref}(t)$ for $t\in\ztm$, where $x_k^\text{ref}$ and $d_k^\text{ref}$, the state and disturbance references, are known. Denoting the nominal state and input with $z_k$ and $v_k$, respectively, the nominal dynamics are
\begin{align}
z_k(t+1)=A(t)z_k(t)+B(t)v_k(t)+d_k^\text{ref}(t),
\label{eq:PREF_nominal_dynamics}
\end{align}
with $z_k(0)=\x$. The nominal dynamics \cref{eq:PREF_nominal_dynamics} are a close approximation of \cref{eq:PF_system} if $z_k \approx x_k^\text{ref}$. The input $u_k(t)$ is designed to ensure that nominal and actual state trajectories remain close by choosing
\begin{align}
u_k(t)=v_k(t)+K(t)(x_k(t)-z_k(t)),%
\label{eq:PREF_true_input}
\end{align}
where $K(t)\in\R^{n_u \times n_x}$ is a time-varying control gain designed so that each $A_K(t):=A(t)+B(t)K(t)$ is Schur. In case this is not possible, i.e., when $(A(t),B(t))$ is not stabilizable, we can choose $K(t)=0$ and utilize $v_k(t)$ without any feedback correction.

There are two issues that need to be addressed: first, the state reference $x_k^\text{ref}$ is not necessarily equal to the true state $x_k$; second, $\f(x_k^\text{ref})$ cannot be measured exacly because of the noise.

\subsubsection{Choosing the reference} The simplest strategy to select the state reference is $x_k^\text{ref}=x_{k-1}$, i.e., choosing the value of the state measured in the previous iteration $k-1$. Although this is effective in practice, one desirable property for the algorithm is non-decreasing nominal performance, which can fail to hold under the static assignment $x_k^\text{ref}=x_{k-1}$.

To achieve non-decreasing nominal performance we allow
%
\begin{align}
d_k^\text{ref}(t) \in \left\{ d_{k-1}^\text{ref}(t),d_{k-1}(t) \right\},%
\label{eq:SWN_binary_constraint}
\end{align}
which can be implemented using $T-1$ binary variables regardless of the iteration number. This binary choice allows the disturbance reference to be equal to either the previous reference $d_{k-1}^\text{ref}(t)$ or the previously measured disturbance $d_{k-1}(t)$. Notice that the choice can be different for each time step. In addition, we set $d_{0}^\text{ref}=d_0$.


\subsubsection{Computing set estimates} The choice of the disturbance reference in \cref{eq:SWN_binary_constraint} and the assignment $d_0^\text{ref}=d_0$ implies that, for every $k\in\Np$, the variable $\f(x_k^\text{ref})$ satisfies
\begin{align*}
\f(x_k^\text{ref}(t),t)=\f(x_j(t),t)=d_j(t) ,
\end{align*}
for some $j < k$. In the remainder of this section, we outline a procedure to produce set estimates for generic disturbances $d_k$, with $k\in\N$; this procedure can then be used to estimate any $f(x_k^\text{ref})$, $k\in\Np$.

Consider an iteration $k\in\N$. We define the \emph{disturbance measurement} $\bar{d}_k\in\Rxtm$ for all $t\in\ztm$ as
\begin{align*}
\bar{d}_k(t):= x_k(t +1 )-A(t) x_k(t)-B(t) u_k(t)=d_k(t)+w_k(t).
\end{align*}
The signal $\bar{d}_k$ can be seen as a noisy measurement of the unknown disturbance $d_k$, and it constitutes our best guess of $d_k$ if we only consider measurements gathered in iteration $k$. Notice that if the system was noise-free, i.e., $w_k=0$, then we would be able to retrieve the exact value of the disturbance, as $\bar{d}_k=d_k$. Since the noise $w_k(t)$ is assumed to belong to $\mathcal{W}$ for all $t\in\ztm$, and $\mathcal{W}$ is symmetric, we have for all $t\in\ztm$ that
\begin{align}
d_k(t) \in \bar{\mathcal{D}}_k(t) := \bar{d}_k(t) \oplus \mathcal{W}. 
\label{eq:SN_disturbance_belong}
\end{align}
The set $\bar{\mathcal{D}}_k(t)$ contains all disturbances that are consistent with the measurements $\bar{d}_k(t)$ and constitutes the best set estimate of $d_k(t)$ obtained using information of iteration $k$. The knowledge of $d_k(t)$ can be improved using measurements collected in iterations prior to $k$ by leveraging \cref{ass:PF_disturbance}. Specifically, we can establish the following.

\begin{lemma}\label{lemma:PREF_dist_set}
Under \cref{ass:PF_disturbance,ass:PF_noise}, we have $d_k(t) \in \mathcal{D}_{k|n}(t)$ for all $k,n\in\N$, $k \leq n$, and $t\in\ztm$, where
\begin{align}
\mathcal{D}_{k|n}(t) := \bigcap_{j=0}^n \left\{ \bar{\mathcal{D}}_j(t) \oplus \mathcal{B}(m(t) \|x_j(t)-x_k(t) \|) \right\}.
\label{eq:PREF_set_estimates}
\end{align}
Moreover
\begin{align}
\scalemath{0.96}{\mathcal{D}_{k|n}(t) = \mathcal{D}_{k|n-1}(t) \cap \left\{ \bar{\mathcal{D}}_n(t) \oplus \mathcal{B}(m(t) \norm{x_{n}(t)-x_{k}(t)}) \right\}.} \label{eq:SN_update_disturbance_set}
\end{align}
\end{lemma}
\begin{proof}
Thanks to \cref{ass:PF_disturbance} we have $d_k(t)\in d_j(t) \oplus \mathcal{B}(m(t) \|  x_j(t)- x_k(t) \|)$ for all $0 \leq j < k$. Therefore
\begin{align}
d_k(t)\in \bar{d}_j(t) \oplus \mathcal{W} \oplus \mathcal{B}(m(t) \|  x_j(t)- x_k(t) \|), \label{eq:SN_dist_lipschitz_rel2}
\end{align}
for all $j \leq k$, since, using \eqref{eq:SN_disturbance_belong}, we have $d_j(t)\in \bar{d}_j(t) \oplus \mathcal{W}$.
Combining \cref{eq:SN_disturbance_belong} and \cref{eq:SN_dist_lipschitz_rel2} we have $d_k(t)\in \left\{ \bar{d}_k(t) \oplus \mathcal{W} \right\} \cap \left\{ \bar{d}_j(t) \oplus \mathcal{W} \oplus \mathcal{B}(m(t) \|  x_j(t)- x_k(t) \|) \right\}$, and repeating for $j=0,1,\dots,n$ we obtain \cref{eq:PREF_set_estimates}. \cref{eq:SN_update_disturbance_set} follows from \cref{eq:PREF_set_estimates}.
\end{proof}
%
%
We refer to $\mathcal{D}_{k|n}(t)$ as \emph{disturbance estimate set} and define $\mathcal{D}_{0|-1}(t)=\bar{d}_0(t)\oplus \mathcal{W}$ for all $t$. \cref{fig:SN_set_membership} showcases a simple example of the intersection in \cref{eq:PREF_set_estimates} for $k=1$.
Similar to $\mathcal{D}_{k|n}$, we define
\begin{align*}
\mathcal{D}_{k|n}^\text{ref}(t) := \bigcap_{j=0}^n \left\{ \bar{\mathcal{D}}_j(t) \oplus \mathcal{B}(m(t) \|x_j(t)-x_k^\text{ref}(t) \|) \right\}.
\end{align*}
with $\mathcal{D}_{0|-1}^\text{ref}(t)=\bar{d}_0(t)\oplus \mathcal{W}$ for all $t$.
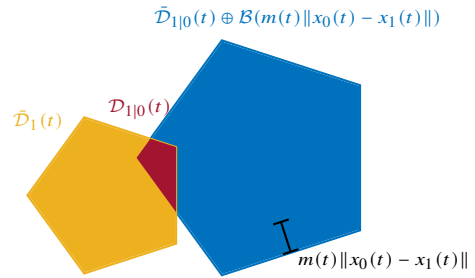
\begin{figure}[h]
\centering
\begin{tikzpicture}

\scriptsize

\tikzstyle{pentagon} = [regular polygon, regular polygon sides=5, draw, thick, inner sep=0, rotate=18]

\definecolor{mycolor1}{rgb}{0.00000,0.44700,0.74100}%
\definecolor{mycolor3}{rgb}{0.63500,0.07800,0.18400}
\definecolor{mycolor2}{rgb}{0.92900,0.69400,0.12500}%

\node[pentagon, fill=mycolor1, opacity=0.4, draw=none, minimum size=3.25cm] (setA) at (1,1) {};
\draw[mycolor1, thick] (setA.corner 1) -- (setA.corner 2) -- (setA.corner 3) -- (setA.corner 4) -- (setA.corner 5) -- cycle;
\node[pentagon,draw=mycolor1,dashed,minimum size=2.15cm] (setA2) at (1,1) {};

\node[pentagon, fill=mycolor2, opacity=0.3, draw=none, minimum size=2.15cm] (setC) at (-1,0.5) {};
\draw[mycolor2, thick] (setC.corner 1) -- (setC.corner 2) -- (setC.corner 3) -- (setC.corner 4) -- (setC.corner 5) -- cycle;

\begin{scope}
    \clip (setA.corner 1) -- (setA.corner 2) -- (setA.corner 3) -- (setA.corner 4) -- (setA.corner 5) -- cycle;
    \clip (setC.corner 1) -- (setC.corner 2) -- (setC.corner 3) -- (setC.corner 4) -- (setC.corner 5) -- cycle;
    \fill[mycolor3, opacity=0.6] (-3,-3) rectangle (3,3);
\end{scope}

\coordinate (midDashed) at ($(setA2.corner 3)!0.5!(setA2.corner 4)$);
\coordinate (midSolid) at ($(setA.corner 3)!0.5!(setA.corner 4)$);

\draw[|-|, thick] (midDashed) -- node[pos=1,yshift=-0.1cm,below,anchor=west] {$m(t)\|x_0(t)-x_1(t)\|$} (midSolid);

\fill[mycolor1] (1,1) circle (1.2pt) node[below] {$\bar{d}_0(t)$};
\fill[mycolor2] (-1,0.5) circle (1.2pt) node[below] {$\bar{d}_1(t)$};

\node[mycolor1] at (1.5,2.85) {$\bar{\mathcal{D}}_{1|0}(t)\oplus \mathcal{B}(m(t)\|x_0(t)-x_1(t)\|)$};
\node[mycolor2] at (-1.95,1.5) {$\bar{\mathcal{D}}_1(t)$};
\node[mycolor3] at (-0.6,1.65) {$\mathcal{D}_{1|0}(t)$};

\end{tikzpicture}
\caption{A depiction of the operation in \cref{eq:PREF_set_estimates} to obtain the smallest set $\mathcal{D}_{1|0}(t)$ containing $d_1(t)$ (in red). The intersection is between two sets $\bar{\mathcal{D}}_1(t)$ (in yellow), and $\bar{\mathcal{D}}_0(t)$ (dashed blue line) with the addition of $\mathcal{B}(m(t)\|x_0(t)-x_1(t)\|)$ to account for the variation of the state.}\label{fig:SN_set_membership}
\end{figure}

\subsection{State Error evolution}
\label{subsection:PREF_error_evolution}
In this section we study the dynamics of the \emph{state error} $e_k(t):=x_k(t)-z_k(t)$. To this end, let $\rho_k\in\R_{\geq 0}^{T}$ be defined so that for all $t\in\ztm$
\begin{align}
\|d-d_k^\text{ref}(t)\| \leq \rho_k(t),~~ \forall d \in \mathcal{D}_{k|k-1}^\text{ref}(t).%
\label{eq:PREF_disturbance_reference_set_estimate1}
\end{align}
Then we have the following.
\begin{lemma}\label{lemma:PREF_error_bound}
Let \cref{ass:PF_initial_condition,ass:PF_disturbance,ass:PF_noise} hold, and let $\eta_k\in\R^T_{\geq 0}$ be defined for some $k\in\N$ and all $t\in\Z_{[1,T-1]}$ as
\begin{align*}
\eta_k (t+1)=\bar{m}(t) \eta_k (t) + m(t) \| z_k(t) - x_k^\text{ref}(t) \| + \rho_k(t) + \bar{w},
\end{align*}
with $\eta_k(0)=r_0$, $\bar{m}(t):=\|A_K(t)\|+m(t)$, and $\rho_k$ satisfying \cref{eq:PREF_disturbance_reference_set_estimate1}. Then $e_k(t)\in \mathcal{B}(\eta_k(t))$.
\end{lemma}
\begin{proof}
Using \cref{eq:PF_system,eq:PREF_nominal_dynamics}  we have
\begin{align*}
& \, e_k(t+1) - w_k(t) \\
= & \, A(t)[x_k(t)-z_k(t)]+B(t)[u_k(t)-v_k(t)]+d_k(t)-d_k^\text{ref}(t), \\
= & \, A_K(t)e_k(t)+d_k(t)-d_k^\text{ref}(t), \\
= & \, A_K(t)e_k(t)+d_k(t)-\f(x_k^\text{ref}(t),t)+\f(x_k^\text{ref}(t),t) - d_k^\text{ref}(t),
\end{align*}
and consequently
\begin{align}
\|e_k(t+1)\| \leq & \, \|A_K(t)\| \|e_k(t)\| + \|d_k(t)-\f(x_k^\text{ref}(t),t)\| \notag \\ & \, + \|\f(x_k^\text{ref}(t),t)-d_k^\text{ref}(t)\| + \bar{w} \label{eq:PREF_norm_error_evolution}
\end{align}
Next, using \cref{ass:PF_disturbance}, we obtain
\begin{align}
& \, \| d_k(t)-\f(x_k^\text{ref}(t),t) \| \notag \\
\leq \, & m(t) \| x_k(t)-x_k^\text{ref}(t) \|, \notag \\
= \, & m(t) \| z_k(t) + e_k(t) - x_k^\text{ref}(t) \|, \notag \\
\leq \, & m(t) \| z_k(t) - x_k^\text{ref}(t) \| + m(t) \| e_k(t) \|. \label{eq:PREF_error_bound_2}
\end{align}
From \cref{eq:PREF_disturbance_reference_set_estimate1} we have $\|\f(x_k^\text{ref}(t),t)-d_k^\text{ref}(t)\| \leq \rho_k(t)$ which, combined with \cref{eq:PREF_norm_error_evolution,eq:PREF_error_bound_2}, yields $\| e_k(t +1 ) \| \leq \bar{m}(t) \| e_k(t) \| +m(t) \| z_k(t) - x_k^\text{ref}(t) \| + \rho_k(t) + \bar{w}$. Since $e_k(0)=x_k(0)-z_k(0)\in \mathcal{B}(r_0)$, we have $\| e_k(t) \| \leq \eta _k(t)$ for all $k \geq 1$ and $t\in\zt$, which completes the proof.
\end{proof}
The variable $\eta_k$ represents the worst-case radius of the uncertainty ball, and we can use it to keep track of the uncertain evolution of the real state. To implement the scheme introduced in \cref{eq:SWN_binary_constraint}, we modify \cref{eq:PREF_disturbance_reference_set_estimate1} as follows
\begin{align}\label{eq:PREF_disturbance_reference_set_estimate2}
\begin{split}
& \norm{d^\text{ref}_{k}(t)-d_1 \alpha_k(t) -d_2 (1-\alpha_k(t))} \leq \rho_{k}(t) , \\
& \forall \, d_1 \in \mathcal{D}_{k-1|k-1}(t), ~ d_2 \in \mathcal{D}_{k-1|k-1}^\text{ref}(t),
\end{split}
\end{align}
where $\alpha_k(t)\in\left\{ 0,1 \right\}$ allows the control scheme to decide which disturbance value to use as a reference. Specifically, \cref{eq:PREF_disturbance_reference_set_estimate2} achieves the purpose of enforcing either $\|d^\text{ref}_{k}(t)-f(x_k^\text{ref}(t),t)\| \leq \rho_k(t)$, when $\alpha_k(t)=0$, or $\|d^\text{ref}_{k}(t)-f(x_k(t),t)\| \leq \rho_k(t)$, when $\alpha_k(t)=1$. Notice that if
\begin{align}\label{eq:PREF_update_reference_disturbance_set}
\mathcal{D}_{k|k-1}^\text{ref}(t) =
\begin{cases}
\mathcal{D}_{k-1|k-1}(t), &\text{if}~\alpha_k(t)=1,\\
\mathcal{D}_{k-1|k-1}^\text{ref}(t), &\text{if}~\alpha_k(t)=0.
\end{cases}
\end{align}
then \cref{eq:PREF_disturbance_reference_set_estimate2} is equivalent to \cref{eq:PREF_disturbance_reference_set_estimate1} for any $t\in\ztm$.

\subsection{Robust constraint satisfaction}
\label{subsection:PREF_noise_term}
The noise is randomly sampled at each timestep from the bounded set $\mathcal{W}$. To ensure safety, we consider the following tightened constraint sets
\begin{align*}
\begin{split}
\mathscr{X}^{\delta(t)} &:= \{ (z,v)\in\R^{n_x+n_u} : H_x z + H_u \leq h - \delta(t) \}, \\
\mathscr{X}_T^{\delta(T)} &:= \{ z\in\R^{n_x} : H_{x,T} z \leq h_T - \delta(T) \}.
\end{split}
\end{align*}
where $\mathscr{Z}^\delta$ is defined according to the previous definition, and $\delta\in\Rp^{Tn_h}$ satisfies
\begin{align}
\label{eq:PREF_robust_constraint_implication}
(z,v) \in \mathscr{Z}^\delta \implies (x,u)\in \mathscr{Z} ,~~~ \forall w\in \mathcal{W}^{T-1},
\end{align}
where $n_h$ denotes the number of rows in $H_x$. Since tightening constraints is usually detrimental for performance, in \cref{section:BLMPC} we devise a strategy to minimize the norm of $\delta$ while fulfilling \cref{eq:PREF_robust_constraint_implication}.
\begin{lemma}\label{lemma:PREF_rob_constr_sat}
Suppose $x_k(t)\in z_k(t) \oplus \mathcal{B}_p(\eta_k(t))$ for $t\in\zt$, then \cref{eq:PREF_robust_constraint_implication} holds if $\delta(t) := \psi(t) \eta(t)$ with
\begin{subequations}
\label{eq:PREF_supports}
\begin{align}
\psi(t) &= \sup_{v\in \mathcal{B}_p(1)} [H_x + H_uK(t)]v,~~ \forall t\in\ztm,\\
\psi(T) &= \sup_{v\in \mathcal{B}_p(1)} [H_{x,T}]v.
\end{align}
\end{subequations}
\end{lemma}
\begin{proof}
Using the definition of $u_k(t)$ in \cref{eq:PREF_true_input} we have
\begin{multline*}
H_x x_k(t) + H_u u_k(t) \in H_x z_k(t) + H_u v_k(t) \\ \oplus \, [H_x + H_uK(t)] \mathcal{B}_p(\eta_k(t)),
\end{multline*}
and we therefore require
\begin{align*}
\delta(t) \geq [H_x + H_uK(t)] v,~~ \forall v \in \mathcal{B}_p(\eta_k(t)),
\end{align*}
which is equivalent to
\begin{align*}
H_x z_k(t) + H_u v_k(t) + \sup_{v\in \mathcal{B}_p(\eta_k(t))} [H_x + H_uK(t)] \leq h,
\end{align*}
where the $\sup$ operation is carried out row-wise. Using \cite[Lemma 2.24 (c)]{beck2017first}, we have that
\begin{align*}
\sup_{v\in \mathcal{B}_p(\eta_k(t))} [H_x + H_uK(t)] = \eta_k(t) \sup_{v\in \mathcal{B}_p(1)} [H_x + H_uK(t)].
\end{align*}
A similar reasonining can be used for $t=T$.
\end{proof}
In the remainder of this paper, we focus on solving the robust problem \Pabsw for the smallest value of $\delta$ satisfying \cref{eq:PREF_robust_constraint_implication}.
\begin{problem}[$\mathcal{P}_\delta$]
    \label{prob:Pabsw}
    \begin{align*}
        \underset{z,v}{\text{minimize}} & \quad \norm{z-r}^2_Q \\
        \text{subject to}
        & \quad z(t+1) = A(t) z(t) + B(t) v(t) + \f(z(t),t), \\
        & \quad (z,v) \in \mathscr{Z}^\delta, \\
        & \quad z(0)=\x.
    \end{align*}
\end{problem}

\section{Iterative learning predictive control}\label{section:BLMPC}
In this section we combine the results derived in \cref{section:PREF} to obtain the iterative learning predictive control (ILPC) scheme. We enforce the nominal dynamics \cref{eq:PREF_nominal_dynamics} and the tightened constraints on the left-hand side of \cref{eq:PREF_robust_constraint_implication}, where the value of $\delta$ is given in \cref{lemma:PREF_rob_constr_sat}, the evolution of the uncertainty $\eta_k$ is given in \cref{lemma:PREF_error_bound} and the value of $\rho_k$ is chosen to satisfy \cref{eq:PREF_disturbance_reference_set_estimate1}. The disturbance reference is chosen dynamically using \cref{eq:SWN_binary_constraint}.

The ILPC scheme involves two separate optimization problems. The first problem, denoted \Pilc, is solved before the beginning of each new iteration $k$. Problem \Pilc is parameterized by the sets $\mathcal{D}_{k-1|k-1},\mathcal{D}_{k-1|k-1}^\text{ref}$, and by the state trajectories $x_{k-1},x_{k-1}^\text{ref}$. In particular, when solving \Pilc for iteration $k$, we optimize over the entire trajectory $(z_k,v_k)$, and therefore do not require online measurements of the state $x_k(t)$ of the system in iteration $k$. \Pilc is a binary mixed-integer program and it is similar to a norm-optimal iterative learning control update (see, e.g., \cite{owens2005iterative}).\par\medskip
\noindent\hypertarget{prob:Pilc}{\textbf{Problem}} $\mathcal{P}_\text{ILC}(r,\mathcal{D}_{k-1|k-1},\mathcal{D}_{k-1|k-1}^\text{ref},x_{k-1},x_{k-1}^\text{ref})$.
\begin{align*}
\min_{\mathbf{x}_k} & \quad \|z_k-r\|_Q^2 + c_1 \|\xi_k\|_1 + \sum_{t=0}^{T-1} c_2(t) \rho_k(t) \\
\text{s.t.}%
& \quad z_k(0)=x_0,~~ \eta_k(0)=r_0, \\
& \quad z_k(t+1)=A(t)z_k(t)+B(t)v_k(t)+d_k^\text{ref}(t),\\
& \quad \eta_k(t+1)=\bar{m}(t)\eta_k(t)+m(t)\xi_k(t) + \rho_k(t) + \bar{w}, \\
& \quad x_k^\text{ref}(t) = \alpha_k(t) x_{k-1}(t) + (1-\alpha_k(t)) x_{k-1}^\text{ref}(t), \\
& \quad \|z_k(t)-x_k^\text{ref}(t)\| \leq \xi_k(t), \\
& \quad \|d^\text{ref}_{k}(t)-d_1 \alpha(t) -d_2 (1-\alpha(t))\| \leq \rho_{k}(t) , \\ 
& \quad \forall \, d_1 \in \mathcal{D}_{k-1|k-1}(t), ~ d_2 \in \mathcal{D}_{k-1|k-1}^\text{ref}(t), \\
& \quad (z_k,v_k)\in\mathscr{Z}^{\psi(t)\eta(t)}, \\
& \quad \alpha_k(t) \in \left\{ 0,1 \right\},~~t\in\ztm,
\end{align*}
where $Q\in\PDset{n_x}$, $c_1>0$ can be taken arbitrarily close to $0$, $c_2(t)=c_1/m(t)$, and $\mathbf{x}_k = (z_k,v_k,\xi_k,\rho_k,\eta_k,d_k^\text{ref},x_k^\text{ref},\alpha_k)$. The auxiliary variable $\xi_k\in\R^{T-1}$ upper-bounds the quantities $\|z_k(t)-x_k^\text{ref}(t)\|$ and avoids introducing unnecessary nonconvexity.

Next, we consider a second problem, denoted \Pmpc, which is solved online, at each time step, after $x_k(t)$ is measured. As the total time duration of the process is finite, this procedure is carried out in a shrinking horizon approach. Problem \Pmpc is similar to \Pilc, with the difference that the binary constraints have been removed, and the state and disturbance references have been replaced with the corresponding outputs of \Pilc. In this way, \Pmpc is a convex optimization problem (specifically, it is a second order cone program if $p=2$ in \cref{ass:PF_disturbance}, and a quadratic program if $p=\infty$) that can be solved efficiently.\par\medskip
\noindent\hypertarget{prob:Pmpc}{\textbf{Problem}} $\mathcal{P}_\text{MPC}(x_k(\t),\t,x_k^{\text{ref},\t},d_k^\text{ref},\rho_k,v_k)$.
\begin{align*}
\min_{\mathbf{x}_k^i} & \quad %
\|z_k^\t-r^\t\|^2_Q + c_1\|\xi_k^\t\|_1 + \sum_{t=0}^{T-\t-1} \| v_k^\t(t)-v_k(\t+t) \|_P^2 \\
\text{s.t.}
& \quad \eta_k^\t(0) \geq \norm{z_k^\t(0)-x_k(\t)}_p,\\
& \quad z_k^\t(t+1)=A(t)z_k^\t(t)+B(t)v_k^\t(t)+d_k(\t+t),\\
& \quad \eta_k^\t(t+1)=\bar{m}(t)\eta_k^\t(t)\!+\!m(t)\xi_k^\t(t)\!+\!\rho_k(\t+t)\!+\!\bar{w},\\
& \quad \norm{z_k^\t(t)-x_k^{\text{ref},\t}(t)}_p \leq \xi_k^\t(t),\\
& \quad (z_k^\t(t),v_k^\t(t))\in\mathscr{X}^{\psi(\t+t)\eta_k^\t(t)},\\
& \quad z_k^\t(T-\t)\in\mathscr{X}_T^{\psi(\t+T)\eta_k^\t(T-\t)},\\
& \quad t\in\Z_{[0,T-\t-1]},
\end{align*}
where $P\in\PSDset{n_u}$ and $\mathbf{x}_k^i=(z_k^\t,v_k^\t,\xi_k^\t,\eta_k^\t,d_k^\t,\rho_k^\t)$. We denote with the superscript $i$ the time-step at which the optimization takes place (e.g. $z_k^i$ denotes time $i$ and iteration $k$, since the total number of time-steps remaining until the end of the iteration is $T-i$, we have $z_k^i\in\R^{(T-i)n_x}$). We use similar notations for $\xi_k^i$, $v_k^i$, $\eta_k^i$, $x_k^{\text{ref},i}$, and $r^i$.

\cref{alg:BMPC} summarizes the control strategy. The non-convex and computationally intensive \Pilc is solved before the execution of the system begins with the goal of identifying the references $x_k^\text{ref},d_k^\text{ref}$. Then, \Pmpc is solved in shrinking horizon to reject disturbances within an iteration while satisfying constraints.

\begin{algorithm}[t!]
\caption{Iterative Learning Predictive Control Scheme}
\label{alg:BMPC}
\begin{algorithmic}[1]
\Init Set $x_0^\text{ref}(t) = x_0(t)$ and $\mathcal{D}^\text{ref}_{0|-1}(t)=\bar{d}_0(t) \oplus \mathcal{W}$ for all $t$.
\While{not converged}
    \State Measure $\bar{d}_{k-1}(t)$ for $t\in\ztm$.
    \State Construct $\mathcal{D}_{k-1|k-1}(t)$ for $t\in\ztm$ using \cref{eq:PREF_set_estimates}.
    \State Construct $\mathcal{D}_{k-1|k-2}^\text{ref}$ for $t\in\ztm$ using \cref{eq:SN_update_disturbance_set}.
    \State Solve \Pilc and store its optimizers $x_k^{\text{ref}}$, $\alpha_k$.
    \State Compute $\mathcal{D}_{k|k-1}^\text{ref}(t)$ for $t\in\ztm$ using \cref{eq:PREF_update_reference_disturbance_set}.
    \For{$\t=0,\dots,T-1$}
    \State Measure $x_k(\t)$.
    \State Set $x_k^{\text{ref},\t}(t)=x_k^{\text{ref},*}(\t+t)$ for $t\in\Z_{[0,T-\t-1]}$.
    \State Solve \Pmpc and store $v_k^i(0)$, $z_k^i(0)$.
    \State Apply $u_k(\t)=v_k^{\t,*}(0)+K(t)(x_k(\t)-z_k^{\t,*}(0))$.
    \State $t\gets t+1$
    \EndFor
    \State $k\gets k+1$
\EndWhile
\end{algorithmic}
\end{algorithm}
The constraint in \cref{eq:PREF_disturbance_reference_set_estimate2} can be implemented efficiently if $\mathcal{D}_{k-1|k-1}(t)$ and $\mathcal{D}_{k-1|k-1}^\text{ref}(t)$ are given in vertex representation. If the number of vertices grows too large as $k$ increases, we can replace $\mathcal{D}_{k-1|k-1}(t)$ and $\mathcal{D}_{k-1|k-1}^\text{ref}(t)$ with over approximations with fixed complexity using specialized software, e.g. \cite{kvasnica2004multi}. How this may hinder the theoretical results of \cref{section:AN} is left for future work.
%

\section{Analysis}\label{section:AN}
In this section, we describe the theoretical properties of \cref{alg:BMPC}.

\subsection{Robust constraint satisfaction}

We begin the analysis by showing that the control inputs $u_k$ obtained generated by \cref{alg:BMPC} produce state trajectories $x_k$ that satisfy the system constraints $(x_k,u_k)\in\mathscr{Z}$ for all iterations and for all possible realizations of the noise $w_k$. For notational simplicity, in the following $\mathbf{x}_k$ denotes an optimizer of \Pilc for iteration $k$, and similarly $\mathbf{x}_k^{\t}$ denotes an optimizer of \Pmpc for iteration $k$ and time $\t$.

We now show that a result similar to \cref{lemma:PREF_error_bound} continues to hold if the variables $x_k(t)$, $z_k(t)$, $\eta_k(t)$ are replaced with $x_k(\t+t)$, $z_k^\t(t)$, $\eta_k^\t(t)$.
\begin{lemma}
\label{lemma:AN_tubes}
Let \cref{ass:PF_disturbance,ass:PF_initial_condition,ass:PF_noise} hold. Then for any $k\in\Np$
\begin{align}
\label{eq:AN_tubes_1}
x_k(t)\in z_k(t) \oplus \mathcal{B}(\eta_k(t))
\end{align}
for $t\in\zt$. Moreover, given any $k\in\Np$ and $\t\in\ztm$, for each $t\in\Z_{[T-\t]}$
\begin{align}
\label{eq:AN_tubes_2}
x_k(\t+t)\in z_k^\t(t) \oplus \mathcal{B}(\eta_k^\t(t)).
\end{align}
\end{lemma}
\begin{proof}
\cref{eq:AN_tubes_1} follows immediately from \cref{lemma:PREF_error_bound}. To prove \cref{eq:AN_tubes_2}, consider that
\begin{multline}
\label{eq:AN_tubes_3}
\|e_k^\t(t+1)\| \leq \bar{m}(\t+t) + m(\t+t) \|z_k^\t(t)-x_k^\text{ref}(\t+t)\| \\ +\rho_k(\t+t)+\bar{w},
\end{multline}
where $e_k^\t(t):=x_k(\t+t)-z_k^\t(t)$.
Moreover, the constraint $\eta_k^\t(0)\geq\|z_k^\t(0)-x_k(\t)\|$ ensures that $x_k(\t)\in z_k^\t(0)\oplus \mathcal{B}(\eta_k^\t(0))$. Combining this with \cref{eq:AN_tubes_3} and the definition of $\eta_k^\t$ in \Pmpc, we obtain $\|e_k^\t(t)\| \leq \eta_k^\t(t)$, proving \cref{eq:AN_tubes_2}.
\end{proof}
%
%
\begin{proposition}[Robust constraint satisfaction]
\label{prop:AN_robust_constraint_satisfaction}
Let \cref{ass:PF_disturbance,ass:PF_initial_condition,ass:PF_noise} hold. Suppose that, for some $k\in\Np$ and for all $\t\in\ztm$, a solution $\mathbf{x}_k^\t$ to \Pmpc exists. Then the control inputs $u_k$ generated via \cref{alg:BMPC} produce closed-loop state trajectories $x_k$ satisfying $(x_k,u_k)\in\mathscr{Z}$.
\end{proposition}
\begin{proof}
For a given time step $\t\in\ztm$, we have from \cref{lemma:AN_tubes} that $x_k(\t)\in z_k^\t(0) \oplus \mathcal{B}_p(\eta_k^\t(0))$. Moreover, since the variables $(z_k^\t(0),v_k^\t(0))$ must satisfy the constraint $(z_k^\t(0),v_k^\t(0))\in \mathscr{X}^{\psi(\t)\eta_k^\t(0)}$, we have that $(x_k(t),u_k(t))\in \mathscr{X}$ thanks to \cref{lemma:PREF_rob_constr_sat}. The same holds for $\t=T-1$ since $x_k(T)\in z_k^{T-1}(1) \oplus \mathcal{B}_p(\eta_k^{T-1}(1))$ and $z_k^{T-1}(1)\in \mathscr{X}_T^{\psi(T)\eta_k^{T-1}(1)}$. This completes the proof.
\end{proof}
\subsection{Recursive feasibility}
Next, we consider \emph{recursive feasibility}. Since the proposed scheme is repetitive, the concept of recursive feasibility is different than the standard definition in the MPC literature (see, e.g., \cite{rakovic_handbook_2018}). Recursive feasibility in this setting means that, assuming problem \Pilc admits a feasible solution at $k=1$, then the problem will be feasible for all $k\in\Np$. In addition, under the same initial feasibility assumption, we require that the problem \Pmpc is feasible for all $t\in\ztm$ and $k\in\Np$.

To ensure feasibility for $k=1$, we must impose an assumption on the initial trajectory $x_0,u_0$, which must satisfy the robust constraints under the worst-case value of $\eta$.
\begin{assumption}
\label{ass:AN_initial_trajectories}
We have access to trajectories $x_0\in\Rxt,u_0\in\Rutm$ that satisfy $(x_0,u_0)\in\mathscr{Z}^\delta$, with $\delta(t)=\psi(t)\eta_0(t)$, $\psi$ is as in \cref{eq:PREF_supports}, $\eta_0(0)=\bar{w}$, and $\eta_0(t+1)=\bar{m}(t)\eta_0(t)+2\bar{w}$ for $t\in\ztm$.
\end{assumption}
\cref{ass:AN_initial_trajectories} is required to guarantee the feasibility of \Pilc for $k=1$, when only one choice of disturbance reference exists (i.e. $d_1^\text{ref}=d_0$).
\begin{proposition}[Recursive feasibility, iteration axis]
\label{prop:AN_recursive_feasibility}
Let \cref{ass:PF_disturbance,ass:PF_initial_condition,ass:PF_noise,ass:AN_initial_trajectories} hold. Then for any $k \geq 1$ problem \Pilc has a feasible solution.
\end{proposition}

\begin{proof}
We use induction. First, for $k=1$ we use the initial trajectory $(x_0,u_0)$ given in \cref{ass:AN_initial_trajectories} to construct the feasible trajectory $\tilde{\mathbf{x}}_1$, defined as follows
\begin{align}
\label{eq:AN_rec_feas_1}
\begin{split}
&\tilde{z}_1 = x_0,~~\tilde{v}_1 = u_0,~~\tilde{d}_1^\text{ref}=\bar{d}_0,~~\tilde{x}_1^\text{ref} = x_0,\\
&\tilde{\alpha}_1=\mathbf{1},~~\tilde{\xi}_1=\mathbf{0},~~\rho_1=\mathbf{1}\bar{w},~~\tilde{\eta}_1(0)=r_0,\\
&\tilde{\eta}_1(t+1)=\tilde{\eta}_1(t)+2\bar{w},~~t\in\ztm.
\end{split}
\end{align}
Trajectory $\tilde{\mathbf{x}}_1$ satisfies all constraints in \Pilc thanks to \cref{ass:AN_initial_trajectories}. Now consider, for any $k \geq 2$, the trajectory $\tilde{\mathbf{x}}_k$ defined as $\tilde{\mathbf{x}}_k=\mathbf{x}_{k-1}$ with the modification $\tilde{\alpha}_k=\mathbf{0}$. It can be verified that $\tilde{\mathbf{x}}_k$ is feasible by construction for \Pilc for iteration $k$. This completes the induction step and the proof.
\end{proof}
Similarly, we can prove that \Pmpc is recursively feasible in the time axis.
\begin{proposition}[Recursive feasibility, time axis]
\label{prop:AN_recursive_feasibility_time}
Let \cref{ass:PF_disturbance,ass:PF_initial_condition,ass:PF_noise,ass:AN_initial_trajectories} hold. Then \Pmpc admits a feasible solution for any $k \geq 1$ and $t\in\ztm$.
\end{proposition}
\begin{proof}
We use induction. First, using \cref{prop:AN_recursive_feasibility}, we know that for any $k \geq 1$ a feasible solution of \Pilc exists. For $t=0$, it can be shown that the trajectory $\tilde{\mathbf{x}}_k^0$, defined as
\begin{align}
\label{eq:AN_rec_feas_BMPC0}
\tilde{z}_k^0=z_k,~~\tilde{v}_k^0=v_k,~~\tilde{\xi}_k^0=\xi_k,~~\tilde{\rho}_k^0=\rho_k,~~\tilde{\eta}_k^0=\eta_k,
\end{align}
is a feasible solution for \Pmpc. Next, let $\tilde{\mathbf{x}}_k^{\t+1}$ be defined for all $t$ as
\begin{align}
\begin{split}
&\tilde{v}_k^{\t+1}(t) = v_k^{\t}(t+1),~~\tilde{\xi}_k^{\t+1}(t) = \xi_k^{\t}(t+1),\\
&\tilde{d}_k^{\t+1}(t)=d_k^{\t}(t+1),~~\tilde{\rho}_k^{\t+1}(t)=\rho_k^{\t}(t+1),\\
&\tilde{z}_k^{\t+1}(t)=z_k^{\t,*}(t+1),
\end{split}%
\label{eq:AN_rec_feas_BMPC1}
\end{align}
and
\begin{align}
\label{eq:AN_rec_feas_BMPC3}
\tilde{\eta}^{\t+1}_k(t+1) = \bar{m}(\t+t+1) \tilde{\eta}^{\t+1}_k(t) + m&(\t+t+1) \tilde{\xi}_k^{\t+1}(t) \notag \\ & ~ + \tilde{\rho}_k^{\t+1}(t) + \bar{w},
\end{align}
with $\tilde{\eta}_k^{\t+1}(0)=\norm{x_k(\t)-z_k^{\t}(1)}_p$. Using \cref{eq:AN_tubes_2} we have $x_k(\t+1) \in z_k^{\t}(1) \oplus \mathcal{B}_p(\eta_k^{\t}(1))$, meaning that $\tilde{\eta}_k^{\t+1}(0) = \|x_k(\t+1)-z_k^{\t+1}(0)\|_p \leq \eta_k^\t(1)$. Combining this fact with the definition of $\tilde{\eta}_k^{\t+1}$ in \cref{eq:AN_rec_feas_BMPC3}, we have $\tilde{\eta}_k^{\t+1}(t) \leq \eta_k^{\t}(t+1)$ for all $t\in\Z_{[0,T-\t]}$. From this we conclude that $\tilde{\mathbf{x}}_k^{\t+1}$ is a feasible solution of \Pmpc. This concludes the proof.
\end{proof}
\subsection{Non-degrading performance}
Next, we show that the optimal cost obtained by \Pilc is non-increasing in $k\in\Np$. We use $\mathcal{J}_\text{ILC}(\mathbf{x}_k)$ and $\mathcal{J}_\text{MPC}(\mathbf{x}_k^\t)$ to denote the cost of \Pilc and \Pmpc attained by the variables $\mathbf{x}_k$ and $\mathbf{x}_k^\t$, respectively.
%
%
\begin{proposition}[Non-degrading performance, iteration axis]
\label{prop:AN_non_deg_it}
Let \cref{ass:PF_disturbance,ass:PF_initial_condition,ass:PF_noise,ass:AN_initial_trajectories} hold. Then for any $k \geq 2$, $\mathcal{J}_\text{ILC}(\mathbf{x}_k) \leq \mathcal{J}_\text{ILC}(\mathbf{x}_{k-1})$. In addition $\mathcal{J}_\text{ILC}(\mathbf{x}_1) \leq \|x_0-r\|_Q^2 + c_1 \|\mathbf{1}\bar{w}\|_1 +c_2 \|\mathbf{1}\bar{w}\|_1$, with $x_0$ satisfying \cref{ass:AN_initial_trajectories}.
\end{proposition}
\begin{proof}
We first prove the second claim of the proposition. Given any state and input trajectory $x_0,u_0$ satisfying \cref{ass:AN_initial_trajectories}, the variables in \cref{eq:AN_rec_feas_1} constitute a feasible solution $\tilde{\mathbf{x}}_1$ to \Pilc for $k=1$. The cost associated to such variables is
\begin{align*}
\mathcal{J}_\text{ILC}(\tilde{\mathbf{x}}_1) = \|x_0-r\|_Q^2 + c_1 \|\mathbf{1}w_\text{max}\|_1 +c_2 \|\mathbf{1}w_\text{max}\|_1.
\end{align*}
By optimality of \Pilc, we conclude that $\mathcal{J}_\text{ILC}(\mathbf{x}_1) \leq \mathcal{J}_\text{ILC}(\tilde{\mathbf{x}}_1)$, proving the second claim of the proposition.

We proceed with the first claim. Given a solution $\mathbf{x}_{k-1}$ of \Pilc for iteration $k-1$, the trajectory $\tilde{\mathbf{x}}_k=\mathbf{x}_{k-1}$ with the modification $\tilde{\alpha}_k=\mathbf{0}$ constitutes a feasible solution to \Pilc for iteration $k$. The cost associated to $\tilde{\mathbf{x}}_k$ is $\mathcal{J}_\text{ILC}(\tilde{\mathbf{x}}_k) = \mathcal{J}_\text{ILC}(\mathbf{x}_{k-1})$. By optimality of $\mathbf{x}_k$, we conclude that $\mathcal{J}_\text{ILC}(\mathbf{x}_k) \leq \mathcal{J}_\text{ILC}(\tilde{\mathbf{x}}_k) = \mathcal{J}_\text{ILC}(\mathbf{x}_{k-1})$, proving the first claim of the proposition and completing the proof.
\end{proof}
Similarly, the cost $\mathcal{J}_\text{MPC}(\mathbf{x}_k^{\t})$ is non-increasing in $\t$ for any $k \geq 1$.
\begin{proposition}[Non-degrading performance, time axis]
\label{prop:AN_non_deg_time}
Let \cref{ass:PF_disturbance,ass:PF_initial_condition,ass:PF_noise,ass:AN_initial_trajectories} hold. Then for any $k \geq 1$ and $\t\in\Z_{[0,T-2]}$
\begin{align*}
\mathcal{J}_\text{MPC}(\mathbf{x}_k^{\t+1}) \leq & \, \mathcal{J}_\text{MPC}(\mathbf{x}_k^{\t})-\|z_k^{\t}(0)-r(\t)\|_Q^2 \\ &-\| v_k{^\t}(0)-v_k(i) \|_P^2 -c_1|\xi_k^{\t}(0)|.
\end{align*}
Moreover
\begin{align*}
\mathcal{J}_\text{MPC}(\mathbf{x}_k^{0}) & \leq \mathcal{J}_\text{ILC}(\mathbf{x}_k)-c_2 \| \rho _k \|_1.
\end{align*}
\end{proposition}
\begin{proof}
The variables in \cref{eq:AN_rec_feas_BMPC1,eq:AN_rec_feas_BMPC3} constitute a feasible solution $\tilde{\mathbf{x}}_k^{\t+1}$ to \Pmpc for iteration $k$ and time $\t+1$ with cost
\begin{align*}
\mathcal{J}_{\text{MPC}}(\tilde{\mathbf{x}}_k^{\t+1}) =~ & %
\sum_{t=0}^{T-\t} \|z_k^{\t}(t+1)\|_Q^2 + c_1 \sum_{t=0}^{T-\t} [ |\xi_k^{\t}(t+1)| \\ & + \|v_k^{\t}(t+1)-v_k(t+\t+1)\|_P^2 ] \\
=~ & \mathcal{J}_\text{MPC}(\mathbf{x}_k^{\t})-\|z_k^{\t}(0)-r(\t)\|_Q^2 \\ & -c_1|\xi_k^{\t}(0)|-c_2|\rho_k^{\t}(0)|.
\end{align*}
By optimality of $\mathbf{x}_k^{\t-1}$, we conclude
\begin{align*}
\mathcal{J}_{\text{MPC}}(\mathbf{x}_k^{\t+1}) \leq \mathcal{J}_{\text{MPC}}(\tilde{\mathbf{x}}_k^{\t+1}) \leq & \, \mathcal{J}_\text{MPC}(\mathbf{x}_k^{\t})-\|z_k^{\t}(0)-r(\t)\|_Q^2 \\ & -c_1|\xi_k^{\t}(0)|-c_2|\rho_k^{\t}(0)|.
\end{align*}
The second statement follows similarly by evaluating the cost associated to the trajectory $\tilde{\mathbf{x}}_k^0$ as defined in \cref{eq:AN_rec_feas_BMPC0}, which is feasible for \Pmpc at iteration $k$ and time $0$. This completes the proof.
\end{proof}

\subsection{Convergence to optimal performance}

Thanks to \cref{prop:AN_non_deg_it}, the sequence $\{\mathcal{J}_\text{ILC}(\mathbf{x}_k)\}_{k\in\N}$ is non-increasing in $k$; therefore, since $\mathcal{J}_\text{ILC}(\mathbf{x}_k)$ is non-negative, the sequence must converge as $k\to\infty$. As a result, the sequence of optimizers $\{z_k,\rho_k,\xi_k\}_{k\in\N}$ is bounded and must admit a convergent subsequence.
We now show that only one fixed point exists, and conclude that the sequence $\{z_k,\rho_k,\xi_k\}_{k\in\N}$ must converge to it.

Before proceeding with the main result of this paper, we require additional assumptions.
\begin{assumption}%
\label{ass:AN_strict_feasibility}
For every $k\in\N$, $(z_k,v_k)\in \operatorname*{int}\mathscr{Z}^\delta_k$, with $\delta_k(t)=\psi(t)\eta_k(t)$ and $\psi$ as in \cref{eq:PREF_supports}.
\end{assumption}
\cref{ass:AN_strict_feasibility} is quite restrictive; nevertheless, it is necessary for the convergence result in \cref{thm:AN_convergence}. In short, \cref{ass:AN_strict_feasibility} requires the asymptotic value of $(z_k,v_k)$ to either lie in the interior of the feasible set, or to approach its boundary while maintaining strict feasibility. A similar assumption is used in \cite{rosolia2017learning} (compare \cite[Theorem 3]{rosolia2017learning}).

For the ease of analysis, we assume that the reference $r$ is realizable. This is consistent with other works in the MPC literature, e.g. \cite[Assumption 4]{kohler2024analysis}.
\begin{assumption}
\label{ass:AN_perfect_tracking}
There exists some $v^*$ and $\eta^*$ such that $(r,v^*,\eta^*)$ solves \Pabsw.
\end{assumption}
Our scheme in \cref{alg:BMPC} is optimistic, since the norm constraint on $d_k^\text{ref}(t)$ in \Pilc enables the scheme to freely choose $d_k^\text{ref}$ to obtain the best nominal cost, at the expense of increased uncertainty. This freedom generally results in better practical performance.

We now present the main convergence result.
\begin{theorem}%
\label{thm:AN_convergence}
Suppose \cref{ass:PF_noise,ass:PF_disturbance,ass:PF_initial_condition,ass:AN_initial_trajectories,ass:AN_perfect_tracking,ass:AN_strict_feasibility} hold. Then with probability $1$
\begin{align*}
\limsup_{k\to \infty} \|z_k(t)-r(t)\| \leq h_t(c_1),
\end{align*}
for all $t\in\zt$ and some $h_t\in\kinf$.
\end{theorem}
\begin{proof}
See Appendix \ref{sec:appendix}.
\end{proof}
\cref{thm:AN_convergence} states that the nominal state trajectory $z_k$ approaches $r$ as $k\to \infty$ up to a factor that depends on $c_1$. For small values of $c_1$, $z_k$ can be made arbitrarily close to $r$. The presence of $c_1$ is necessary to avoid unnecessarily large values of the term $\rho_k$. Practically, as showcased in the simulation example, one can expect excellent tracking performance even for non-trivial values of $c_1$, suggesting that the presence of the terms depending on $c_1$ may be an artifact of the way the proofs are structured, and may be removed at the cost of more involved proofs, or more restrictive assumptions. This is however beyond the scope of this paper, and will be addressed in future work.

\section{Binary constraint relaxation}
If the disturbance term $d_k(t)$ is an affine function of $x_k(t)$, \Pilc can be reformulated as a convex optimization problem while retaining all theoretical properties described in \cref{section:AN}. Specifically, let
\begin{align}
\label{eq:AD_disturbance}
d_k(t) := f(x_k(t),t) = D(t) x_k(t) + d(t),
\end{align}
where $D(t)\in\R^{n_x \times n_x}$ with $\|D(t)\|_p \leq m(t)$, and $d(t)\in\R^{n_x}$ for $t\in\Z_{[0,T-1]}$. By leveraging the linearity of $f(x,t)$ in $x$, we can replace the integer constraint $\alpha_k(t)\in\left\{ 0,1 \right\}$ with the convex constraint $\alpha_k(t)\in[0,1]$, rendering \Pilc a quadratic program (if $p=\infty$). To see why this is possible, we show that relation \cref{eq:AN_tubes_1} is still verified if $\alpha_k(t)\in[0,1]$. Showing that the scheme fulfills the theoretical properties of \cref{section:AN} is straightforward and therefore omitted.
\begin{proposition}
Suppose $d_k(t)$ is given as in \cref{eq:AD_disturbance}, and that the constraint $\alpha_k(t)\in\left\{ 0,1 \right\}$ in \Pilc is replaced with $\alpha_k(t)\in[0,1]$ for all $t\in\Z_{[0,T-1]}$. Then \cref{eq:AN_tubes_1} is still verified for all $k \geq 1$ and $t\in\Z_{[T-1]}$.
\end{proposition}
\begin{proof}
Let $\mathcal{D}_{k-1|k-1}(t)=\hat{d}_{k-1}(t)\oplus \mathcal{B}(\rho_{k-1}(t))$ and $\mathcal{D}_{k-1|k-1}^\text{ref}(t)=\hat{d}_{k-1}^\text{ref}(t)\oplus \mathcal{B}(\rho_{k-1}^\text{ref}(t))$. Then we have
\begin{align*}
& \f(\alpha_k(t) x_{k-1}(t) + (1-\alpha_k(t)) x_{k-1}^\text{ref}(t),t) \\
& \qquad = D(t)(\alpha_k(t) x_{k-1}(t) + (1-\alpha_k(t)) x_{k-1}^\text{ref}(t)) + d(t) \\
& \qquad = \alpha_k(t) \f(x_{k-1}(t),t) + (1-\alpha_k(t)) \f(x_{k-1}^\text{ref}(t),t),
\end{align*}
for all $t\in\Z_{[0,T_1]}$. As a result, we have
\begin{align*}
&\alpha_k(t) \f(x_{k-1}(t),t) + (1-\alpha_k(t)) \f(x_{k-1}^\text{ref}(t),t) \\
\in~& \alpha_k(t) \hat{d}_{k-1}(t) + (1-\alpha_k(t)) \hat{d}_{k-1}^\text{ref}(t) \\ & \oplus \mathcal{B}( \alpha_k(t)\rho_{k-1}(t)+(1-\alpha_k(t))\rho_{k-1}^\text{ref}(t)).
\end{align*}
Therefore, for any $k \geq 1$ and $t\in\Z_{[0,1]}$ we have
\begin{align*}
\f(x_k^\text{ref}(t),t)=~&\f(\alpha_k(t) x_{k-1}(t) + (1-\alpha_k(t)) x_{k-1}^\text{ref}(t),t) \\
\in~& \alpha_k(t) \mathcal{D}_{k-1|k-1}(t) + (1-\alpha_k(t)) \mathcal{D}_{k-1|k-1}^\text{ref}(t),
\end{align*}
and the result follows from the proof of \cref{lemma:PREF_error_bound}.
\end{proof}

\section{Simulation example}\label{section:SIM}
We deploy our scheme on a batch process described by time-invariant dynamics given as in \cref{eq:PF_system} with
\begin{align*}
A(t)=\begin{bmatrix} 1.3070&1\\-0.6086&0 \end{bmatrix} ,~~~ B(t)= \begin{bmatrix} 1.239\\-0.8282 \end{bmatrix},
\end{align*}
with initial condition $x_k(0)\in(-0.5,0.5)\oplus \mathcal{B}(\bar{w})$ for all $k \geq 0$, where $\bar{w}=0.06$ and we chose $p=\infty$. The horizon is chosen as $N=30$. The constraints are given by $\| x_k(t) \|_\infty \leq x_\text{max}=1.75$, $|u_k (t)| \leq u_\text{max}=0.85$. The reference $r(t)$ is computed by first applying the input $\bar{u}(t)\in\R^{n_u}$ for $t\in[0,T \!-\!1 ]$ given in \cref{fig:SIM_reference_input_signal}, to the system under the assumption that no disturbance is present, i.e. $\f(x,t)=0$ for all $x,t$. The resulting trajectory $r(t)$ can be seen in \cref{fig:SIM_reference_signal}. The disturbance is given by the quadratic function
\begin{align*}
[\f(x,t)]_i = \alpha x_i^2 + [d(t)]_i,
\end{align*}
where the constant $\alpha = m/{x_\text{max}^2}$ ensures that $\f$ satisfies \cref{ass:PF_disturbance} for all feasible states, and $d(t)$ is given as in \cref{fig:SIM_state_independent_disturbance}. We choose $m=0.2$.
\begin{figure}[b]
\centering
%
%
\definecolor{mycolor1}{rgb}{0.4940 0.1840 0.5560}
\begin{tikzpicture}

\begin{axis}[%
width=\pgfscale*4.392in,
height=\pgfscale*0.6*1.952in,
at={(\pgfscale*0.737in,\pgfscale*0.316in)},
scale only axis,
xmin=0,
xmax=30,
xlabel style={font=\color{white!15!black}},
xlabel={Time-step $t$},
ymin=-0.11,
ymax=0.11,
ylabel style={font=\color{white!15!black}},
ylabel={$\bar{u}$},
axis background/.style={fill=white},
title style={font=\bfseries},
axis x line*=bottom,
axis y line*=left,
xmajorgrids,
ymajorgrids,
xtick distance = 5,
label style={font=\small},
tick label style={font=\small},
scaled y ticks=base 10:1,
ytick = {-0.1,0,0.1},
]
\addplot [color=mycolor1, dashed, line width=1.0pt, mark size=1pt, mark=*, mark options={solid, mycolor1}, forget plot]
  table[row sep=crcr]{%
1	-0.1\\
2	-0.1\\
3	-0.1\\
4	-0.1\\
5	-0.1\\
6	-0.1\\
7	-0.1\\
8	-0.1\\
9	-0.1\\
10	-0.1\\
11	-0.1\\
12	-0.1\\
13	-0.1\\
14	-0.1\\
15	-0.1\\
16	-0.1\\
17	-0.1\\
18	-0.1\\
19	0.1\\
20	0.1\\
21	0.1\\
22	0.1\\
23	0.1\\
24	0.1\\
25	0\\
26	0\\
27	0\\
28	0\\
29	0\\
};
\end{axis}
\end{tikzpicture}%
\caption{Reference input signal $\bar{u}(t)$.}
\label{fig:SIM_reference_input_signal}
\end{figure}
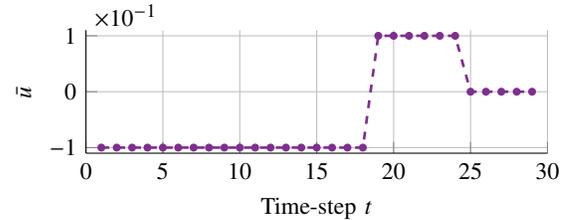
\begin{figure}
\centering
%
%
\definecolor{mycolor1}{rgb}{0.00000,0.44700,0.74100}%
\definecolor{mycolor2}{rgb}{0.85000,0.32500,0.09800}
%
\begin{tikzpicture}

\begin{axis}[%
width=\pgfscale*4.521in,
height=\pgfscale*1.659in,
at={(\pgfscale*0.758in,\pgfscale*0.481in)},
scale only axis,
xmin=0,
xmax=30,
xlabel style={font=\color{white!15!black}},
xlabel={Time-step $t$},
ylabel style={font=\color{white!15!black}},
ylabel={$r$},
axis background/.style={fill=white},
axis x line*=bottom,
axis y line*=left,
xmajorgrids,
ymajorgrids,
legend style={legend columns = 2, legend cell align=left, align=left, draw=white!15!black,legend style={at={(0,2)},anchor=west}},
xtick distance = 5,
label style={font=\small},
tick label style={font=\small}
]
\addplot [color=mycolor1, dashed, line width=1.0pt, mark size=1pt, mark=*, mark options={solid, mycolor1}]
  table[row sep=crcr]{%
1	-0.5\\
2	-0.2774\\
3	-0.0993417999999999\\
4	-0.0020940925999999\\
5	0.0166424404518001\\
6	-0.0180538655731374\\
7	-0.0748049915630561\\
8	-0.127862541385103\\
9	-0.162670023725054\\
10	-0.175872578321671\\
11	-0.171944483427357\\
12	-0.158775388672986\\
13	-0.143954020381704\\
14	-0.132597203092507\\
15	-0.126774127637602\\
16	-0.126075127020246\\
17	-0.128705456935217\\
18	-0.132568709909807\\
19	-0.136017162761344\\
20	0.109626885122031\\
21	0.267142384111049\\
22	0.323516173747873\\
23	0.301332784118485\\
24	0.238030005499905\\
25	0.168794084773865\\
26	-0.00707119254779995\\
27	-0.111970128653349\\
28	-0.142041430365336\\
29	-0.117503129189066\\
30	-0.0671301753297658\\
} node [below,pos=0.5,yshift=-0.15cm] {$[r(t)]_1$} ;
\addplot [color=mycolor2, dashed, line width=1.0pt, mark size=1pt, mark=*, mark options={solid, mycolor2}]
  table[row sep=crcr]{%
1	0.5\\
2	0.38712\\
3	0.25164564\\
4	0.14327941948\\
5	0.0840944647563599\\
6	0.0726914107410345\\
7	0.0938075825878114\\
8	0.128346317865276\\
9	0.160637142686974\\
10	0.181820976439068\\
11	0.189856051166569\\
12	0.187465412613889\\
13	0.179450701546379\\
14	0.170430416804305\\
15	0.1635186578021\\
16	0.159974734080245\\
17	0.159549322304522\\
18	0.161150141090773\\
19	0.163501316851108\\
20	-3.99547434457581e-05\\
21	-0.149538922285268\\
22	-0.245402854969984\\
23	-0.279711943342955\\
24	-0.26621113241451\\
25	-0.227685061347242\\
26	-0.102728079993375\\
27	0.00430352778459105\\
28	0.0681450202984283\\
29	0.0864464145203436\\
30	0.0715124044244656\\
}node [above,pos=0.5,yshift=0.15cm] {$[r(t)]_2$} ;

\end{axis}
\end{tikzpicture}%
\caption{Reference signal $r(t)$.}
\label{fig:SIM_reference_signal}
\end{figure}
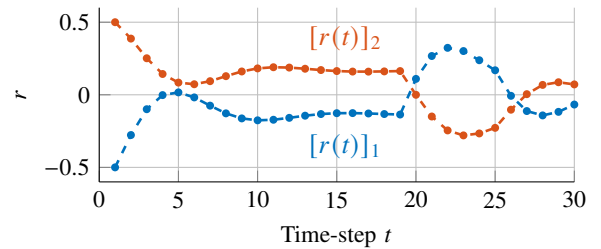
\begin{figure}
\centering
%
%
\definecolor{mycolor1}{rgb}{0.00000,0.44700,0.74100}%
\definecolor{mycolor2}{rgb}{0.85000,0.32500,0.09800}%
\begin{tikzpicture}



\begin{axis}[%
width=\pgfscale*4.521in,
height=\pgfscale*1.659in,
at={(\pgfscale*0.758in,\pgfscale*0.481in)},
scale only axis,
xmin=0,
xmax=30,
xlabel style={font=\color{white!15!black}},
xlabel={Time-step $t$},
ylabel style={font=\color{white!15!black}},
ylabel={$d$},
axis background/.style={fill=white},
axis x line*=bottom,
axis y line*=left,
xmajorgrids,
ymajorgrids,
legend style={legend cell align=left, align=left, draw=white!15!black},
xtick distance = 5,
label style={font=\small},
tick label style={font=\small}
]
\addplot [color=mycolor1, dashed, line width=1.0pt, mark size=1pt, mark=*, mark options={solid, mycolor1}]
  table[row sep=crcr]{%
1	0.2478\\
2	0.2478\\
3	0.2478\\
4	0.3717\\
5	0.3717\\
6	0.3717\\
7	0.2478\\
8	0.2478\\
9	0.2478\\
10	0.2478\\
11	0.2478\\
12	0.2478\\
13	0.4956\\
14	0.4956\\
15	0.4956\\
16	0.1239\\
17	0.1239\\
18	0.1239\\
19	0.1239\\
20	0.1239\\
21	0.1239\\
22	0.4956\\
23	0.4956\\
24	0.4956\\
25	0.3717\\
26	0.3717\\
27	0.3717\\
28	0.2478\\
29	0.2478\\
} node [above,pos=0.62,yshift=0.15cm] {$[d(t)]_1$} ;
\addplot [color=mycolor2, dashed, line width=1.0pt, mark size=1pt, mark=*, mark options={solid, mycolor2}]
  table[row sep=crcr]{%
1	-0.16564\\
2	-0.16564\\
3	-0.16564\\
4	-0.24846\\
5	-0.24846\\
6	-0.24846\\
7	-0.16564\\
8	-0.16564\\
9	-0.16564\\
10	-0.16564\\
11	-0.16564\\
12	-0.16564\\
13	-0.33128\\
14	-0.33128\\
15	-0.33128\\
16	-0.08282\\
17	-0.08282\\
18	-0.08282\\
19	-0.08282\\
20	-0.08282\\
21	-0.08282\\
22	-0.33128\\
23	-0.33128\\
24	-0.33128\\
25	-0.24846\\
26	-0.24846\\
27	-0.24846\\
28	-0.16564\\
29	-0.16564\\
} node [below,pos=0.62,yshift=-0.15cm] {$[d(t)]_2$} ;

\end{axis}
\end{tikzpicture}%
\caption{State-independent disturbance $d(t)$.}
\label{fig:SIM_state_independent_disturbance}
\end{figure}
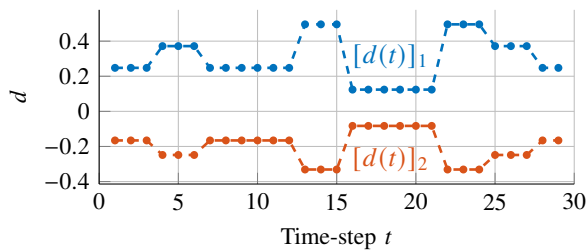
The control scheme uses $Q=I$ and $c_1=0.1$. The ancillary controller gain $K$ is chosen as
\begin{align*}
K = &\, \argmin_{\tilde{K}} \norm{A+B \tilde{K}} \\
& \text{s.t.} ~~ |\lambda| < 1 , ~~ \forall \lambda \in \sigma(A),
\end{align*}
where $\sigma(A)$ is the spectrum of $A$. We have 
\begin{align*}
K = \begin{bmatrix}-0.9075&-0.5029\end{bmatrix} ,~~ \| A\!+\!BK \|_{\infty} = 0.5595.
\end{align*}
Since $m=0.2$, we have $\bar{m}=\| A\!+\!BK \|_{i \infty } + m = 0.7595$. The initial trajectory $x_0,u_0$ is generated by using a simple feedback control law $u_0(t) = K_0(x_0(t)-r(t)),~~t\in\ztm$, where $K_0$ is chosen to place the closed loop poles of $A+BK_0$ to $[0.3,0.2]$.

\Cref{fig:SIM_state_tracking} shows the closed-loop trajectories for different values of $k$. The line in orange represents the initial trajectory $x_0$ generated with the feedback controller. The line in blue shows $x_{50}$, and the dashed grey lines represent the trajectories $x_{5k}$ with $k=2,\dots,9$. As shown in the plots, the ILPC scheme is able to greatly improve upon the initial trajectory, achieving almost perfect tracking of the reference (in red).

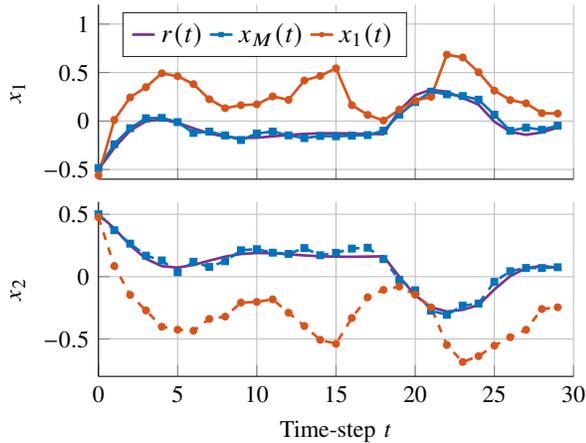
\begin{figure}
\centering
%
%
\definecolor{mycolor1}{rgb}{0.4940 0.1840 0.5560}
\definecolor{mycolor2}{rgb}{0.00000,0.44700,0.74100}%
\definecolor{mycolor3}{rgb}{0.85000,0.32500,0.09800}%
\begin{tikzpicture}

\begin{axis}[%
width=\pgfscale*4.521in,
height=\pgfscale*1.662in,
at={(\pgfscale*0.758in,\pgfscale*2.36in)},
scale only axis,
xmin=0,
xmax=30,
ymin=-0.6,
ymax=1.2,
ylabel style={font=\color{white!15!black}},
ylabel={$x_1$},
axis background/.style={fill=white},
title style={font=\bfseries},
axis x line*=bottom,
axis y line*=left,
xmajorgrids,
ymajorgrids,
legend style={at={(0.35,0.7)}, anchor=south, legend columns=3, legend cell align=left, align=left, draw=white!15!black},
legend image post style={scale=0.5},
separate axis lines,
xtick distance = 5,
xticklabels = empty,
label style={font=\small},
tick label style={font=\small}
]
\addplot [color=mycolor1, line width=1.0pt]
  table[row sep=crcr]{%
0	-0.5\\
1	-0.2774\\
2	-0.0993417999999999\\
3	-0.0020940925999999\\
4	0.0166424404518001\\
5	-0.0180538655731374\\
6	-0.0748049915630561\\
7	-0.127862541385103\\
8	-0.162670023725054\\
9	-0.175872578321671\\
10	-0.171944483427357\\
11	-0.158775388672986\\
12	-0.143954020381704\\
13	-0.132597203092507\\
14	-0.126774127637602\\
15	-0.126075127020246\\
16	-0.128705456935217\\
17	-0.132568709909807\\
18	-0.136017162761344\\
19	0.109626885122031\\
20	0.267142384111049\\
21	0.323516173747873\\
22	0.301332784118485\\
23	0.238030005499905\\
24	0.168794084773865\\
25	-0.00707119254779995\\
26	-0.111970128653349\\
27	-0.142041430365336\\
28	-0.117503129189066\\
29	-0.0671301753297658\\
};
\addlegendentry{$r(t)$}

\addplot [color=mycolor2, line width=1.0pt, mark size=1pt, mark=square*, mark options={solid, mycolor2}]
  table[row sep=crcr]{%
0	-0.485391205232388\\
1	-0.239163550340206\\
2	-0.0740906188510623\\
3	0.0282561764473951\\
4	0.0322066708476182\\
5	-0.0111746428214023\\
6	-0.122673020885011\\
7	-0.109090647835487\\
8	-0.14879030611282\\
9	-0.195992078326456\\
10	-0.128587431286966\\
11	-0.108642872192381\\
12	-0.147740326074401\\
13	-0.176205470698283\\
14	-0.154124735131131\\
15	-0.157165058654132\\
16	-0.151100924119038\\
17	-0.145212608069383\\
18	-0.0995067036066172\\
19	0.0668422256533813\\
20	0.19579977309656\\
21	0.303073634362294\\
22	0.276114593092534\\
23	0.257786399364435\\
24	0.220244490043164\\
25	0.0650772877884096\\
26	-0.10085357746002\\
27	-0.0679836185154162\\
28	-0.0902745735738788\\
29	-0.04751440492731\\
};
\addlegendentry{$x_M(t)$}

\addplot [color=mycolor3, line width=1.0pt, mark size=1pt, mark=*, mark options={solid, mycolor3}]
  table[row sep=crcr]{%
0	-0.559986275021919\\
1	0.0112488501352189\\
2	0.243495199991331\\
3	0.349887484273902\\
4	0.492601172718387\\
5	0.461918849324848\\
6	0.380592233535361\\
7	0.224849671360519\\
8	0.132489636659112\\
9	0.163462431524253\\
10	0.171999632385732\\
11	0.253876255576853\\
12	0.218441303624811\\
13	0.41818728888397\\
14	0.465705142475881\\
15	0.543717639431203\\
16	0.163690748799429\\
17	0.0632513577160313\\
18	0.0052302296392092\\
19	0.117918944359216\\
20	0.209352106199514\\
21	0.249219424875257\\
22	0.684707843608308\\
23	0.65558310544784\\
24	0.503876560006641\\
25	0.313535186836197\\
26	0.216648117361268\\
27	0.182752184254432\\
28	0.0815504314515054\\
29	0.0770603854071963\\
};
\addlegendentry{$x_1(t)$}

\end{axis}

\begin{axis}[%
width=\pgfscale*4.521in,
height=\pgfscale*1.662in,
at={(\pgfscale*0.758in,\pgfscale*0.481in)},
scale only axis,
xmin=0,
xmax=30,
xlabel style={font=\color{white!15!black}},
xlabel={Time-step $t$},
ymin=-0.8,
ymax=0.6,
ylabel style={font=\color{white!15!black}},
ylabel={$x_2$},
axis background/.style={fill=white},
axis x line*=bottom,
axis y line*=left,
xmajorgrids,
ymajorgrids,
legend style={at={(0.5,0.97)}, anchor=north, legend columns=3, legend cell align=left, align=left, draw=white!15!black},
xtick distance = 5,
label style={font=\small},
tick label style={font=\small}
]
\addplot [color=mycolor1, line width=1.0pt]
  table[row sep=crcr]{%
0	0.5\\
1	0.38712\\
2	0.25164564\\
3	0.14327941948\\
4	0.0840944647563599\\
5	0.0726914107410345\\
6	0.0938075825878114\\
7	0.128346317865276\\
8	0.160637142686974\\
9	0.181820976439068\\
10	0.189856051166569\\
11	0.187465412613889\\
12	0.179450701546379\\
13	0.170430416804305\\
14	0.1635186578021\\
15	0.159974734080245\\
16	0.159549322304522\\
17	0.161150141090773\\
18	0.163501316851108\\
19	-3.99547434457581e-05\\
20	-0.149538922285268\\
21	-0.245402854969984\\
22	-0.279711943342955\\
23	-0.26621113241451\\
24	-0.227685061347242\\
25	-0.102728079993375\\
26	0.00430352778459105\\
27	0.0681450202984283\\
28	0.0864464145203436\\
29	0.0715124044244656\\
};

\addplot [color=mycolor2, dashed, line width=1.0pt, mark size=1pt, mark=square*, mark options={solid, mycolor2}]
  table[row sep=crcr]{%
0	0.500389626509504\\
1	0.372157855496243\\
2	0.265074141240345\\
3	0.168387879452131\\
4	0.130084349434359\\
5	0.0354615587268484\\
6	0.119646987911878\\
7	0.0778723699034915\\
8	0.123998250118225\\
9	0.211776754052544\\
10	0.220524305863556\\
11	0.192254637488174\\
12	0.182756337136304\\
13	0.229232151903869\\
14	0.171835863970182\\
15	0.191092641979912\\
16	0.225045205723102\\
17	0.232382000130238\\
18	0.140908479950685\\
19	-0.0258408419367319\\
20	-0.109562875718062\\
21	-0.272741965874909\\
22	-0.305677288164521\\
23	-0.232805332945535\\
24	-0.215819843026137\\
25	-0.0416090927619638\\
26	0.043793917756102\\
27	0.0696501042218182\\
28	0.0689899126450936\\
29	0.075433741416102\\
};

\addplot [color=mycolor3, dashed, line width=1.0pt, mark size=1pt, mark=*, mark options={solid, mycolor3}]
  table[row sep=crcr]{%
0	0.476279908715821\\
1	0.0851656172315571\\
2	-0.146551381610347\\
3	-0.270923928480405\\
4	-0.401888681407485\\
5	-0.425172238740543\\
6	-0.433627170304453\\
7	-0.33886670332339\\
8	-0.320550681272706\\
9	-0.208102880604605\\
10	-0.202637317837997\\
11	-0.182019644863518\\
12	-0.288992403858148\\
13	-0.395916873856446\\
14	-0.506787270752982\\
15	-0.538056302486612\\
16	-0.331709901969151\\
17	-0.16550701605297\\
18	-0.105985226476198\\
19	-0.079287793752088\\
20	-0.144581962641619\\
21	-0.245823112067633\\
22	-0.547941221035777\\
23	-0.68442171879295\\
24	-0.637636641427002\\
25	-0.553514537620878\\
26	-0.485977948635115\\
27	-0.426031406295588\\
28	-0.259715480840283\\
29	-0.245711217508464\\
};

\end{axis}
\end{tikzpicture}%
\caption{Closed-loop state trajectories.}
\label{fig:SIM_state_tracking}
\end{figure}

\cref{fig:SIM_comparison_noise_nonlinear} compares the closed loop and the open loop costs of the ILPC with the cost obtained by applying the open-loop input $u_k(t)=v_k^*(t)+K(x_k(t)-z_k^*(t))$ (i.e. solving only \Pilc and not repeating the optimization at every time-step in shrinking horizon fashion). The plots shows also the costs of an MPC that utilizes $d_k^\text{ref}=d_1$, where $d_1$ is assumed to be known exactly. Clearly, the BMPC scheme is able to improve upon the performance of both the ILC scheme and the MPC scheme by utilizing the shrinking horizon strategy and by updating $d_k^\text{ref}$ at the beginning of each new iteration, respectively.

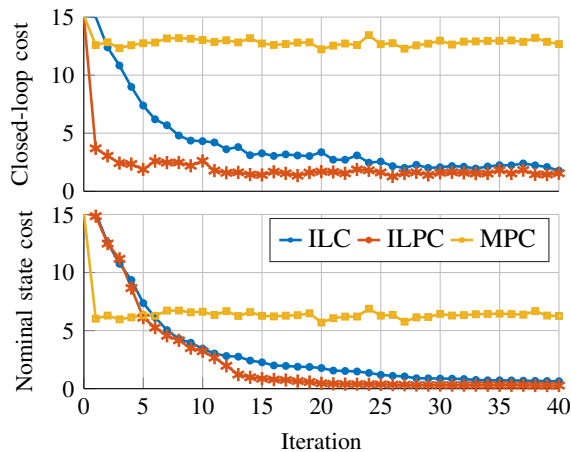
\begin{figure}
\centering
%
%
\definecolor{mycolor1}{rgb}{0.00000,0.44700,0.74100}%
\definecolor{mycolor2}{rgb}{0.85000,0.32500,0.09800}%
\definecolor{mycolor3}{rgb}{0.92900,0.69400,0.12500}%
\begin{tikzpicture}

\begin{axis}[%
width=\pgfscale*4.521in,
height=\pgfscale*1.659in,
at={(\pgfscale*0.758in,\pgfscale*2.363in)},
scale only axis,
xmin=0,
xmax=40,
ymin=0,
ymax=15,
ylabel style={font=\color{white!15!black}},
ylabel={Closed-loop cost},
axis background/.style={fill=white},
title style={font=\bfseries},
axis x line*=bottom,
axis y line*=left,
xmajorgrids,
ymajorgrids,
legend style={legend cell align=left, align=left, draw=white!15!black},
separate axis lines,
xtick distance = 5,
xtick style={draw=none},
xticklabels = empty,
x axis line style = { draw opacity=0},
label style={font=\small},
tick label style={font=\small}
]
\addplot [color=mycolor1, line width=1.0pt, mark size=1pt, mark=*, mark options={solid, mycolor1},forget plot]
  table[row sep=crcr]{%
0	15.1463570256496\\
1	15.0259838309112\\
2	12.3937118638146\\
3	10.829954054437\\
4	8.99297992694262\\
5	7.38434723593837\\
6	6.19744167038528\\
7	5.68660119167429\\
8	4.80183965111926\\
9	4.36467647633443\\
10	4.31563512831839\\
11	4.20711821076084\\
12	3.62061039996489\\
13	3.80697956119786\\
14	3.10735548039213\\
15	3.27136666106238\\
16	3.04453450353994\\
17	3.176767863974\\
18	3.08416665037228\\
19	3.03569049048969\\
20	3.36679664991138\\
21	2.72157062115906\\
22	2.71566289058195\\
23	3.08839878138017\\
24	2.46441966450614\\
25	2.56224172673964\\
26	2.16521245918078\\
27	2.01779867901618\\
28	2.2775885518488\\
29	2.01161041705451\\
30	2.07054313005243\\
31	2.17193379500878\\
32	2.11593047693021\\
33	1.9631597278173\\
34	2.13364361465426\\
35	2.24689162023473\\
36	2.23593715498341\\
37	2.39795935863678\\
38	2.23671587035834\\
39	2.09131913017322\\
40	1.77571367617503\\
41	1.88067937651913\\
42	1.8963426625247\\
43	2.05147498469557\\
44	1.79315232071066\\
45	1.73373981839867\\
46	1.97597077642126\\
47	1.91879879237869\\
48	1.81253212782926\\
49	2.04094280955651\\
};
\addplot [color=mycolor2, line width=1.0pt, mark size=2.5pt, mark=asterisk, mark options={solid, mycolor2},forget plot]
  table[row sep=crcr]{%
0	15.1463570920645\\
1	3.6952136623854\\
2	3.04940148910685\\
3	2.43040221229384\\
4	2.34786924738797\\
5	1.88588078747681\\
6	2.60117940142282\\
7	2.472114833334\\
8	2.47726008988776\\
9	2.19848813781522\\
10	2.61676548076703\\
11	1.78355266490359\\
12	1.58953737699046\\
13	1.62926467346491\\
14	1.43652673512561\\
15	1.40827318758636\\
16	1.67924491990686\\
17	1.54404230719429\\
18	1.35134894703197\\
19	1.60902058243078\\
20	1.69140652231635\\
21	1.66185527199223\\
22	1.54749826999176\\
23	1.87676114926188\\
24	1.80615090988131\\
25	1.62208019750875\\
26	1.27905241063899\\
27	1.53651007009412\\
28	1.63362535595399\\
29	1.40119457125166\\
30	1.5904074640631\\
31	1.63485075459569\\
32	1.58589396119655\\
33	1.52252676070772\\
34	1.5299333874234\\
35	1.80978870905007\\
36	1.53691476850918\\
37	1.82538897401454\\
38	1.45831891695144\\
39	1.44207886115284\\
40	1.55108672883689\\
41	1.43798997550716\\
42	1.60969194466676\\
43	1.91531336645967\\
44	1.43612930374054\\
45	1.38777368582835\\
46	1.4294753363526\\
47	1.54384994307331\\
48	1.3483211924991\\
49	1.38805377057389\\
};
\addplot [color=mycolor3, line width=1.0pt, mark size=1pt, mark=square*, mark options={solid, mycolor3},forget plot]
  table[row sep=crcr]{%
0	15.1463570256496\\
1	12.5965178771535\\
2	12.8358795205891\\
3	12.3278590502036\\
4	12.579532978402\\
5	12.749760074004\\
6	12.8120952333828\\
7	13.1674186350547\\
8	13.1899519787689\\
9	13.1375100079446\\
10	13.0290327592444\\
11	12.8615331665601\\
12	13.0037345213626\\
13	12.8215041238857\\
14	13.1827598957097\\
15	12.7291476940783\\
16	12.599117346263\\
17	12.6792742330021\\
18	12.8020681572146\\
19	12.8216198959167\\
20	12.2231710115344\\
21	12.5413696983582\\
22	12.7244485780788\\
23	12.5945864758557\\
24	13.4433350047765\\
25	12.6673893966217\\
26	12.7520619740359\\
27	12.2781881578766\\
28	12.5622916149708\\
29	12.7115150787439\\
30	12.9739505039686\\
31	12.6363146981819\\
32	12.8805731203841\\
33	12.9088801206891\\
34	12.9444537268604\\
35	12.9554256314292\\
36	12.9873963837164\\
37	12.8667089227271\\
38	13.208130670472\\
39	12.8702947779855\\
40	12.6735845334636\\
41	12.903095943719\\
42	13.1525987780866\\
43	13.3073207133427\\
44	12.8783675326478\\
45	12.8139929954652\\
46	12.9612213070364\\
47	12.7116807968181\\
48	12.5630176329636\\
49	12.9380862918812\\
};
\end{axis}

\begin{axis}[%
width=\pgfscale*4.521in,
height=\pgfscale*1.659in,
at={(\pgfscale*0.758in,\pgfscale*0.481in)},
scale only axis,
xmin=0,
xmax=40,
xlabel style={font=\color{white!15!black}},
xlabel={Iteration},
ymin=0,
ymax=15,
ylabel style={font=\color{white!15!black}},
ylabel={Nominal state cost},
axis background/.style={fill=white},
legend style={legend columns=3,legend cell align=left, align=left, draw=white!15!black},
legend image post style={scale=0.5},
axis x line*=bottom,
axis y line*=left,
xmajorgrids,
ymajorgrids,
xtick distance = 5,
label style={font=\small},
tick label style={font=\small}
]
\addplot [color=mycolor1, line width=1.0pt, mark size=1pt, mark=*, mark options={solid, mycolor1}]
  table[row sep=crcr]{%
0	15.1463570256496\\
1	15.0929967698923\\
2	12.5817609742951\\
3	10.7539669197173\\
4	9.35685261865524\\
5	7.3674743247128\\
6	6.14487247420468\\
7	5.03784565817616\\
8	4.29069028568678\\
9	3.95670188058747\\
10	3.44615308302812\\
11	3.0255312849775\\
12	2.81239038362899\\
13	2.75713232877404\\
14	2.41801176244919\\
15	2.26896722615998\\
16	2.00191657509161\\
17	1.95974026753137\\
18	1.8894897672852\\
19	1.87061834182819\\
20	1.77848908235475\\
21	1.55868205828996\\
22	1.52532629770729\\
23	1.48103140698623\\
24	1.35961050442145\\
25	1.19102218636784\\
26	1.10438528322348\\
27	1.05619967308959\\
28	0.897451724680915\\
29	0.888149360495246\\
30	0.868178419629159\\
31	0.861085546085927\\
32	0.830263892786088\\
33	0.750057299213987\\
34	0.717910535354829\\
35	0.717563422132365\\
36	0.711458661407357\\
37	0.68897960757238\\
38	0.668824474586668\\
39	0.657554104532505\\
40	0.636509739358136\\
41	0.62959637287477\\
42	0.61965077912955\\
43	0.603227920313522\\
44	0.580434542154029\\
45	0.57942916725337\\
46	0.575942481315595\\
47	0.575942482862273\\
48	0.575113139974873\\
49	0.575113126698885\\
};
\addlegendentry{ILC}

\addplot [color=mycolor2, line width=1.0pt, mark size=2.5pt, mark=asterisk, mark options={solid, mycolor2}]
  table[row sep=crcr]{%
0	15.1463570920645\\
1	14.8536513205963\\
2	12.49287032617\\
3	11.1812552031333\\
4	8.66304402123951\\
5	6.11677489244614\\
6	5.26063951627928\\
7	4.55507649162457\\
8	4.16593415823588\\
9	3.49160081653113\\
10	3.22390367388169\\
11	2.66777824676639\\
12	1.96777839647957\\
13	1.21368606727147\\
14	0.980845788651661\\
15	0.863060143920061\\
16	0.783755369719838\\
17	0.745805269531216\\
18	0.654904881595136\\
19	0.575208195910453\\
20	0.463514948879598\\
21	0.417187679241636\\
22	0.395033531269992\\
23	0.390863078005501\\
24	0.378398115918893\\
25	0.374700900057002\\
26	0.351169403121764\\
27	0.338221775533119\\
28	0.315104335602632\\
29	0.314145477844165\\
30	0.313953980105451\\
31	0.310787754297136\\
32	0.309648527773344\\
33	0.307493431606799\\
34	0.29879138874899\\
35	0.296205073654725\\
36	0.294717324129127\\
37	0.29414511102578\\
38	0.272401616302404\\
39	0.271177317713774\\
40	0.27068446400365\\
41	0.262527942535541\\
42	0.262171296746235\\
43	0.236451907598903\\
44	0.235757682030309\\
45	0.235757682030656\\
46	0.235100330897558\\
47	0.232066880737124\\
48	0.230876278512446\\
49	0.226878779092702\\
};
\addlegendentry{ILPC}

\addplot [color=mycolor3, line width=1.0pt, mark size=1pt, mark=square, mark options={solid, mycolor3}]
  table[row sep=crcr]{%
0	15.1463570256496\\
1	6.02512856759914\\
2	6.29529801771682\\
3	5.9749128096568\\
4	6.14416981490648\\
5	6.31010578987367\\
6	6.27136728091676\\
7	6.72521104105808\\
8	6.72473153984773\\
9	6.57857508407965\\
10	6.62119259685615\\
11	6.34588255739541\\
12	6.67873581887861\\
13	6.25973626506548\\
14	6.59495801883015\\
15	6.27306546382593\\
16	6.2261367653341\\
17	6.28360271525403\\
18	6.33024245056238\\
19	6.48958569814084\\
20	5.70125544872734\\
21	6.0713476784145\\
22	6.20396985681232\\
23	6.21023472318458\\
24	6.88554382363376\\
25	6.28975031190235\\
26	6.35062827801199\\
27	5.7671182805718\\
28	6.14453812513711\\
29	6.16768055164487\\
30	6.43808202788913\\
31	6.29818970150174\\
32	6.33298220277397\\
33	6.40794406629514\\
34	6.43634996316316\\
35	6.4628621425243\\
36	6.42308786826765\\
37	6.3700329731671\\
38	6.68197342061533\\
39	6.29574121692671\\
40	6.26325695654189\\
41	6.4683359410027\\
42	6.57928139696708\\
43	6.75710668812927\\
44	6.25794922311783\\
45	6.3511780012454\\
46	6.44227969659823\\
47	6.15645472065239\\
48	6.0516855775181\\
49	6.3947809849495\\
};
\addlegendentry{MPC}

\end{axis}
\end{tikzpicture}%
\caption{Comparison of closed-loop cost.}
\label{fig:SIM_comparison_noise_nonlinear}
\end{figure}

\section{Conclusion}\label{section:CONCL}
In this paper we developed an Iterative Learning Predictive Control scheme that is able to control a constrained nonlinear system performing a repeated finite-time operation subject to both bounded additive stochastic noise and additive uncertainty. The uncertainty is assumed to be state-dependent and to satisfy a Lipschitz inequality. The scheme has guaranteed constraint satisfaction and its nominal performance is non-increasing in the iteration axis. Constraint satisfaction is guranteed by constructing $p$-norm balls that are guaranteed to contain the true state of the system. By leveraging the dependency of the uncertainty on the state, the size of the ball is dynamically chosen during the optimization to minimize the conservativeness. The scheme converges to a solution where the uncertain component has been fully learned and the ball size is at its minimum.

The proposed scheme can at times be conservative, because of the well-known conservativeness of Lipschitz bounds. In the future, we may address this issue by considering different kinds of uncertainty propagation, for example by utilizing polytopes or zonotopes. The conservatism could be further reduced by utilizing a probabilistic description of the noise and the uncertainty, as opposed to the currently used deterministic, worst-case description.

\bibliographystyle{IEEEtran}
\bibliography{Sources/ref.bib}

\appendix

\subsection{Proof of \texorpdfstring{Theorem 1}{Theorem \ref{thm:AN_convergence}}}\label{sec:appendix}
The proof of \cref{thm:AN_convergence} requires several technical Lemmas. The proof is based on induction; hence most Lemmas will begin with an condition which must be verified for all time-steps $t\leq i$ (or $t<i$) for some $i\in\Z_{[0,T-1]}$ and demonstrate the condition continues to hold at time $i+1$ (or $i$). The overall reasoning is summarized as follows:
\begin{enumerate}
    \item \cref{lemma:AN_continuity_lemma} proves that if the nominal ILC trajectory remains close to the reference up to some time-step $i$, and the magnitude $\|w_k(t)\|$ of the noise trajectory $w_k(t)$ is small enough, then the real closed-loop trajectory is also close to the reference up to time-step $i-1$.
    \item \cref{lemma:AN_technical_lemma} is a technical Lemma required in the last step of the main proof to obtain a conclusion that holds with probability $1$.
    \item \cref{lemma:AN_disturbance_set_to_singleton} leverages the previous two Lemmas to show that the disturbance estimate set will shrink to an arbitrarily small ball around the true disturbance.
    \item \cref{lemma:AN_xi_and_rho_are_zero} uses \cref{lemma:AN_disturbance_set_to_singleton} to prove that $\xi_k(t)$ and $\rho_k(t)$ will both converge to an arbirarily small neighborhood of the origin.
    \item \cref{prop:AN_thm_intermediate} finally shows that the nominal ILC trajectory will eventually be arbitrarily close to the reference.
    \item \cref{thm:AN_convergence} can then be proven easily using \cref{prop:AN_thm_intermediate}.
\end{enumerate}
We begin by proving that the closed loop state trajectory $x_k(t)$ matches the open loop prediction $z_k(t)$ as provided by solving \Pilc, so long as the noise $w_k(j)$ is zero for $j < t$, and the state $z_k$ coincides with the reference $r$ up to time step $t$. Moreover, leveraging the outer Lipschitz continuity of the set of solutions of quadratic programs subject to linear perturbations of the constraints, we show that $x_k(t)$ remains arbitrarily close to $z_k(t)$ if $\|w_k(t)\|$ is small enough.

\begin{lemma}
\label{lemma:AN_continuity_lemma}
Let $i\in\Z_{[0,T]}$ be fixed. Suppose that
\begin{subequations}\label{eq:AN_cont_lemma_hp}
\begin{align}
\|z_k(t)-r(t)\| &\leq h_1(\c)+\sigma_1(1/k), ~~ \forall t\leq i, \label{eq:AN_cont_lemma_hp_1} \\
\rho_k(t) &\leq h_2(\c)+\sigma_2(1/k), ~~ \forall t<i, \label{eq:AN_cont_lemma_hp_2} \\
\xi_k(t) &\leq h_3(\c)+\sigma_3(1/k), ~~ \forall t<i, \label{eq:AN_cont_lemma_hp_3}
\end{align}
\end{subequations}
for some $h_1,h_2,h_3,\sigma_1,\sigma_2,\sigma_3\in\kinf$. Then
\begin{enumerate}
    \item[(\emph{i})] there exist $h_4,\sigma_4\in\kinf$ such that if $x_k(0)=\x$ and $w_k(t)=0$ for $t<i$, then
    \begin{align}\label{eq:AN_cont_lemma_result}
    \|x_k(i)-r(i)\|\leq h_4(\c)+\sigma_4(1/k).
    \end{align}
    \item[(\emph{ii})] For every $\epsilon>h_4(\c)+\sigma_4(1/k)$ there exists some $\delta>0$ such that if $\|x_k(0)-\bar{x}\|\leq\delta$ and $\|w_k(t)\|\leq\delta$ for all $t<i$, then $\|x_k(i)-r(i)\|\leq\epsilon$.
\end{enumerate}
\end{lemma}
%
%
\begin{proof}[Proof (i)]
The proof of part (i) will proceed as follows
\begin{enumerate}
    \item we work towards a contradiction by assuming that \cref{eq:AN_cont_lemma_result} does not hold for at least one time-step $t$;
    \item next, we bound $\|z_k(t+\ell-1)-r(t+\ell-1)\|_Q^2-\|z_k^{t-1}(\ell)-r(t+\ell-1)\|_Q^2$ for all $\ell\in\Z_{[0,T-t]}$ by separately bounding $\|z_k(t+\ell-1)-r(t+\ell-1)\|+\|z_k^{t-1}(\ell)-r(t+\ell-1)\|$ (leveraging the compactness of $\mathscr{X}_c$) and $\|z_k(t+\ell-1)-r(t+\ell-1)\|-\|z_k^{t-1}(\ell)-r(t+\ell-1)\|$ (by constructing a feasible trajectory for the MPC at time-step $t$ and leveraging \cref{ass:AN_strict_feasibility});
    \item by studying the difference of the ILC cost and the MPC cost (pointwise at each time-step), we then obtain a bound on $\|z_k^t-1(i)-z_k(t-1+i)\|$ for $i=0,1$;
    \item finally, we bound $\|x_k(t)-r(t)\|$.
\end{enumerate}
Suppose that $w_k(t)=0$ for $t<i$ and $x_k(0)=\bar{x}$. Suppose, for the sake of contradiction, that \cref{eq:AN_cont_lemma_result} does not hold for one or more time-step strictly smaller than $i$, let $t<i$ be the smallest such time-step, then \cref{eq:AN_cont_lemma_result} for all $j < t$, that is, there exists some some $\gamma_1,g_1\in\kinf$ such that for all $j<t$
\begin{align}\label{eq:AN_cont_lemma_contr_hp}
\|x_k(j)-r(j)\|\leq\gamma_1(\c)+g_1(1/k).
\end{align}
Since by assumption $x_k(0)=\bar{x}=r(0)$, we have $t>0$.
We now show that for all $\ell\in\Z_{[0,T-t]}$
\begin{multline}
\label{eq:AN_cont_lemma_cost_condition_4}
\|z_k(t+\ell-1)-r(t+\ell-1)\|_Q\\-\|z_k^{t-1}(\ell)-r(t+\ell-1)\|_Q\leq\gamma_2(\c),
\end{multline}
for some $\gamma_2\in\kinf$. To this end, given any $\ell\in\Z_{[0,T-t]}$, consider a candidate solution $\tilde{\mathbf{x}}=\mathbf{x}_k$ for \Pilc with the following modifications
\begin{align*}
\tilde{z}(t\!+\!\ell\!-\!1) & = z_k^{t-1}(\ell ),\\
\tilde{d}(t\!+\!\ell\!-\!2) & = d_k^\text{ref}(t\!+\!\ell\!-\!2)-z_k(t\!+\!\ell\!-\!1)+z_k^{t-1}(\ell),\\
\tilde{d}(t\!+\!\ell\!-\!1) & = d_k^\text{ref}(t\!+\!\ell\!-\!1) - A(t\!+\!\ell\!-\!1)[z_k^{t-1}(\ell) - z_k(t\!+\!\ell\!-\!1)].
\end{align*}
Thanks to \cref{ass:AN_strict_feasibility}, we have that $\theta\tilde{\mathbf{x}}+(1-\theta)\mathbf{x}_k$ is feasible for \Pilc for some $\theta\in[0,1]$ sufficiently small. Observe that
\begin{align}\label{eq:AN_cont_lemma_aux}
\begin{split}
\tilde{\rho}(t\!+\!\ell\!-\!2) & \leq \rho(t\!+\!\ell\!-\!2) + \|z_k^{t-1}(\ell)-z_k(t\!+\!\ell\!-\!1)\|,\\
\tilde{\rho}(t\!+\!\ell\!-\!1) & \leq \rho(t\!+\!\ell\!-\!1) + \theta \|A(t\!+\!\ell\!-\!1)z_k^{t-1}(\ell) - z_k(t\!+\!\ell\!-\!1)\|,\\
\tilde{\xi}(t\!+\!\ell\!-\!1) & \leq \xi(t\!+\!\ell\!-\!1) + \theta \|z_k^{t-1}(\ell)-z_k(t\!+\!\ell\!-\!1)\|.
\end{split}
\end{align}
Since $\theta\tilde{\mathbf{x}}+(1-\theta)\mathbf{x}_k$ and $\mathbf{x}_k$ only differ at time-steps $t+\ell-2$ and $t+\ell-1$, the difference between their cost between is
\begin{align*}\allowdisplaybreaks
&\mathcal{J}_\text{ILC}(\theta\tilde{\mathbf{x}}+(1-\theta)\mathbf{x}_k)-\mathcal{J}_\text{ILC}(\mathbf{x}_k)\\
\stackrel{\cref{eq:AN_cont_lemma_aux}}{\leq}\,&\|\theta z_k^{t-1}(\ell)+(1-\theta)z_k(t+\ell-1)-r(t+\ell-1)\|_Q^2\\ &-\|z_k(t+\ell-1)-r(t+\ell-1)\|_Q^2\\ &+\theta(c_1+c_2)\|z_k^{t-1}(\ell)-z_k(t+\ell-1)\|\\ &+\theta c_2\|A(t+\ell-1)[z_k^{t-1}(\ell)-z_k(t+\ell-1)]\|\\
\stackrel{(a)}{\leq}\,&\|\theta[z_k^{t-1}(\ell)-r(t+\ell-1)]+(1-\theta)[z_k(t+\ell-1)\\ &-r(t+\ell-1)]\|_Q^2-\|z_k(t+\ell-1)-r(t+\ell-1)\|_Q^2\\ &+\theta\bar{k}\|z_k^{t-1}(\ell)-z_k(t\!+\!\ell\!-\!1)\|_Q\\
\stackrel{(b)}{\leq}\,&\theta\|z_k^{t-1}(\ell)-r(t\!+\!\ell\!-\!1)\|_Q^2-\theta\|z_k(t\!+\!\ell\!-\!1)-r(t\!+\!\ell\!-\!1)\|_Q^2\\ &+\theta\bar{k}\|z_k^{t-1}(\ell)-z_k(t\!+\!\ell\!-\!1)\|_Q\\   
\stackrel{(c)}{\leq}\,&\theta[\|z_k^{t-1}(\ell)-r(t+\ell-1)\|_Q+\|z_k(t+\ell-1)\\ &-r(t+\ell-1)\|_Q][\|z_k^{t-1}(\ell)-r(t+\ell-1)\|_Q\\ &-\|z_k(t+\ell-1)-r(t+\ell-1)\|_Q]+\theta\bar{k}\|z_k^{t-1}(\ell)\\ &-r(t+\ell-1)\|_Q+\theta\bar{k}\|z_k(t+\ell-1)-r(t+\ell-1)\|_Q,
\end{align*}
where in $(a)$ we defined
\begin{align*}
\bar{k}&:= k_{p,-} (c_1 + c_2 + c_2 \|A(t+\ell-1)\|) \\
&=c_1k_{p,-}(1+1/m(t+\ell-2)+\|A(t+\ell-1)\|/m(t+\ell-1)),
\end{align*}
where $k_{p,-}>0$ is a constant satisfying $\|v\| \leq k_{p,-}\|v\|_2$ for any $v\in\R^{n_x}$, in $(b)$ we used the identity $\|(1-\theta)a+\theta b\|_2^2=(1-\theta)\|a\|_2^2+\theta \|b\|_2^2-\theta(1-\theta)\|a-b\|_2^2$ (compare equation (4) in \cite{ryu2016primer}), and in $(c)$ we used the relation $a^2-b^2=(a-b)(a+b)$. Defining
\begin{align*}
\epsilon_1:=\|z_k^{t-1}(\ell)-r(t\!+\!\ell\!-\!1)\|_Q-\|z_k(t\!+\!\ell\!-\!1)-r(t\!+\!\ell\!-\!1)\|_Q,
\end{align*}
a sufficient condition for $\mathcal{J}_\text{ILC}(\theta\tilde{\mathbf{x}}+(1-\theta)\mathbf{x}_k)-\mathcal{J}_\text{ILC}(\mathbf{x}_k)<0$ is therefore
\begin{align}
\label{eq:AN_cont_lemma_cost_condition_5}
\bar{k}+\epsilon_1<0\iff\epsilon_1<-\bar{k}.
\end{align}
If \cref{eq:AN_cont_lemma_cost_condition_5} holds then $\theta\tilde{\mathbf{x}}+(1-\theta)\mathbf{x}_k$ solves \Pilc while achieving a smaller cost than $\mathbf{x}_k$, contradicting the optimality of $\mathbf{x}_k$. We therefore conclude that \cref{eq:AN_cont_lemma_cost_condition_5} must not hold. In other words, we need $\epsilon_1 \geq -\bar{k}$, which is equivalent to \cref{eq:AN_cont_lemma_cost_condition_4} for $\gamma_2(\c)=\bar{k}$. Since $\ell\in\Z_{[0,T-t]}$ was chosen arbitrarily, \cref{eq:AN_cont_lemma_cost_condition_4} holds for all $\ell\in\Z_{[0,T-t]}$.

Next, observing that
\begin{align*}
\|z_k(t\,+&\,\ell-1)-r(t+\ell-1)\|_Q\\&+\|z_k^{t-1}(\ell)-r(t+\ell-1)\|_Q \leq 2\lambda_\text{max}(Q) \operatorname{diam}\mathscr{X}_x,
\end{align*}
and that for any $a,b\in\R$ it holds $a^2-b^2=(a-b)(a+b)$, we get
\begin{multline}
\label{eq:AN_cont_lemma_cost_condition_6}
\|z_k(t+\ell-1)-r(t+\ell-1)\|_Q^2-\|z_k^{t-1}(\ell)-r(t+\ell-1)\|_Q^2\\\leq 2\lambda_\text{max}(Q) \operatorname{diam}\mathscr{X}_x \gamma_2(\c)=:\gamma_3(\c).
\end{multline}
We now study the difference in cost (pointwise in time) $\Delta \mathcal{J}(\ell)$ between the ILC solution and the MPC solution for a given $\ell$
\begin{align*}
\Delta \mathcal{J}(\ell) & \stackrel{\hphantom{(27)}}{:=} \|z_k(t+\ell-1)-r(t+\ell-1)\|_Q^2-\|z_k^{t-1}(\ell)\\ &~~~-r(t+\ell-1)\|_Q^2 + c_1\xi_k(t+\ell-1) - c_1 \xi_k^{t-1}(\ell) \\
& \stackrel{\cref{eq:AN_cont_lemma_cost_condition_6}}{\leq} \gamma_3(\c) + c_1 k_{p,-} \|z_k(t+\ell-1)-x_k^\text{ref}(t+\ell-1)\| \\
& \stackrel{\hphantom{(27)}}{\leq} \gamma_3(\c) + c_1 k_{p,-} \operatorname{diam} \mathscr{X}_x =: \gamma_4(\c),
\end{align*}
with $\gamma_4\in\kinf$. Condition \cref{eq:AN_cont_lemma_cost_condition_6} ensures that if $\Delta \mathcal{J}(0)>0$, that is, the \Pmpc solution improves the cost for $\ell=0$ compared to the ILC, the improvement is bounded above by $\gamma_4(c_1)$. By optimality of the solution of \Pmpc we must have $\sum_{i=1}^{T-t}\Delta \mathcal{J}(i)>0$. Therefore, if $\Delta \mathcal{J}(0)<0$, there must be a time-index $\ell>0$ for which $\Delta \mathcal{J}(0)>0$. $\Delta \mathcal{J}(0)$ is most negative when $\Delta \mathcal{J}(0)>0$ for all $\ell\in\Z_{[1,T-t]}$. Therefore we must have
\begin{align*}
& \Delta \mathcal{J}(0) \\
=\,&\|z_k^{t-1}(0)-r(t-1)\|_Q^2 - \|z_k(t-1)-r(t-1)\|_Q^2 \\ & + c_1\|z_k^{t-1}(0)-x_k^\text{ref}(t-1)\|-c_1\|z_k(t-1)-x_k^\text{ref}(t-1)\| \\ \leq & \sum_{\ell=1}^{T-t} \Delta \mathcal{J}(\ell) \leq (T-t)\gamma_4(\c),
\end{align*}
which is equivalent to
\begin{align*}
& \, \|z_k^{t-1}(0)-r(t-1)\|_Q^2+c_1\|z_k^{t-1}-x_k^\text{ref}(t-1)\|  \\
\stackrel{\hphantom{(23a)}}{\leq} & (T-t)\gamma_4(\c) + \|z_k(t\!-\!1)-r(t\!-\!1)\|_Q^2\\ &\,+c_1\|z_k(t\!-\!1)-x_k^\text{ref}(t\!-\!1)\| \\
\stackrel{\cref{eq:AN_cont_lemma_hp_3}}{\leq} & \, (T-t)\gamma_4(\c) +  \lambda_\text{max}^2(Q)k_{p,+}^2\operatorname{diam}\mathscr{X}_x[h_1(\c)+\sigma_1(1/k)]\\&\,+c_1[h_3(\c)+\sigma_3(1/k)]\\
=:&\gamma_5(\c)+g_5(1/k),
\end{align*}
where $k_{p,+}>0$ is a constant satisfying $\|v\| \geq k_{p,+}\|v\|_2$ for all $v\in\R^{n_x}$, and $\gamma_5,g_5\in\kinf$. In particular we have $\|z_k^{t-1}(0)-r(t-1)\|_Q^2\leq\gamma_5(\c)+g_5(1/k)$, which implies
\begin{align*}
\|z_k^{t-1}(0)-r(t-1)\| &\leq \lambda_\text{min}(Q)^{-1}k_{p,-} \sqrt{\gamma_5(\c)+g_5(1/k)} \\
& \leq \lambda_\text{min}(Q)^{-1}k_{p,-} \sqrt{\gamma_5(\c)}+\sqrt{g_5(1/k)}\\
& =: \gamma_6(\c)+g_6(1/k),
\end{align*}
with $\gamma_6,g_6\in\kinf$. Using \cref{eq:AN_cont_lemma_hp_1} and the triangle inequality we have
\begin{align}
\|z_k^{t-1}(0)-z_k(t\!-\!1)\| &\leq h_1(c_1)+\gamma_1(1/k)+\gamma_6(c_1)+g_6(1/k) \notag \\
&=: \gamma_7(c_1)+g_7(1/k). \label{eq:AN_cont_lemma_ilc_mpc_bound_0}
\end{align}
Since $t<i$, the same holds for the subsequent time-step, i.e.,
\begin{align}
\|z_k^{t-1}(1)-z_k(t)\| \leq \gamma_7(c_1)+g_7(1/k). \label{eq:AN_cont_lemma_ilc_mpc_bound_1}
\end{align}
Finally, consider that
{\allowdisplaybreaks
\begin{align}
&\,\|x_k(t)-r(t)\|\notag\\
\stackrel{(a)}{=}&\,\|A(t\!-\!1)x_k(t\!-\!1)-A(t\!-\!1)r(t\!-\!1)+B(t-1)v_k^{t\!-\!1}(0)\notag\\ &+B(t\!-\!1)K(t\!-\!1)[x_k(t\!-\!1)-z_k^{t-1}(0)]-B(t\!-\!1)v^*(t\!-\!1)\notag\\ &+f(x_k(t\!-\!1),t\!-\!1)-f(r(t\!-\!1),t\!-\!1)\|\notag\\
\stackrel{(b)}{\leq}&\,\bar{m}(t-1)\|x_k(t-1)-r(t-1)\|+\|B(t-1)[v_k^{t-1}(0)\notag\\ &-v^*(t-1)]\|+\|B(t-1)K(t-1)[r(t-1)-z_k^{t-1}(0)]\|\notag\\
\stackrel{(c)}{\leq}&\,\bar{m}(t\!-\!1)[\gamma_1(\c)+s_1(1/k)]+\epsilon_2\notag\\&\,+\|B(t\!-\!1)K(t\!-\!1)[r(t\!-\!1)\!-\!z_k^{t-1}(0)]\|\notag\\
\stackrel{\hphantom{(d)}}{\leq}&\,\bar{m}(t\!-\!1)[\gamma_1(\c)+s_1(1/k)]+\epsilon_2+\|B(t\!-\!1)K(t\!-\!1)[r(t\!-\!1)\notag\\ &-z_k(t-\!1)]\|+\|B(t\!-\!1)K(t\!-\!1)[z_k(t\!-\!1)-z_k^{t-1}(0)]\|\notag\\
\stackrel{(d)}{\leq}&\,\bar{m}(t\!-\!1)[\gamma_1(\c)+s_1(1/k)]+\epsilon_2+\|B(t\!-\!1)K(t\!-\!1)\|\notag\\&\,\cdot[h_1(\c)+\sigma_1(1/k)+\gamma_7(\c)+g_7(1/k)],\label{eq:AN_cont_lemma_xk_m_r_bound}
\end{align}}
where $(a)$ follows from the definition of $x_k(t)$ and \cref{ass:AN_perfect_tracking}, $(b)$ follows from \cref{ass:PF_disturbance} and $\bar{m}(t-1)=\|A(t-1)\|+m(t-1)$, in $(c)$ we defined $\epsilon_2:=\|B(t-1)[v_k^{t-1}(0)-v^*(t-1)]\|$ and used \cref{eq:AN_cont_lemma_contr_hp}, and in $(d)$ we used \cref{eq:AN_cont_lemma_hp_1} and \cref{eq:AN_cont_lemma_ilc_mpc_bound_0}. It remains to show that $\epsilon_2\leq\gamma_8(\c)+g_8(1/k)$ for some $\gamma_8,g_8\in\kinf$. To see this, observe that
{\allowdisplaybreaks
\begin{align*}
&\,\|B(t-1)[v_k^{t-1}(0)-v^*(t-1)]\|\\
\stackrel{(a)}{\leq}&\,\|A(t-1)z_k(t-1)+B(t-1)v_k(t-1)+d_k^\text{ref}(t-1)\\ &-A(t-\!1)r(t-\!1)-B(t-\!1)v^*(t-\!1)-f(r(t-\!1),t-\!1)\|\\ &+\|A(t-1)z_k^{t-1}(0)+B(t-1)v_k^{t-1}-A(t-1)z_k(t-1)\\ &-B(t-1)v_k(t-1)\|+\|A(t-1)\|\|z_k^{t-1}(0)-r(t-1)\|\\ &+\|f(z_k(t-1),t-1)-f(r(t-1),t-1)\|\\ &+\|d_k^\text{ref}(t-1)-f(z_k(t-1),t-1)\|\\
\stackrel{(b)}{\leq}&\, \|z_k(t)-r(t)\|+\|z_k^{t-1}(1)-z_k(t)\|+\|A(t-\!1)\|\\ &\,\|z_k^{t-1}(0)-r(t-1)\|+m(t-1)\|z_k(t-1)-r(t-1)\|\\ &+\|d_k^\text{ref}(t-1)-f(z_k(t-1),t-1)\|\\
\stackrel{(c)}{\leq}&\,h_1(\c)+\sigma_1(1/k)+\gamma_7(\c)+g_7(1/k)+\|A(t-1)\|[\gamma_6(\c)\\&\,+g_6(1/k)]+m(t-1)[h_1(\c)+\sigma_1(1/k)]+h_2(\c)\\&\,+\sigma_2(1/k)+m(t-1)h_3(\c)+\sigma_3(1/k)\\
=:&\,\gamma_8(c)+g_8(1/k),
\end{align*}}
where $(a)$ follows by adding and subtracting terms and the submultiplicative property of the $p$-norm, in $(b)$ we used the definition of $z_k(t)$, $z_k^{t-1}(1)$, and $r(t)$, as well as \cref{ass:PF_disturbance}, and in $(c)$ we used \cref{eq:AN_cont_lemma_hp}, \cref{eq:AN_cont_lemma_ilc_mpc_bound_0}, and \cref{eq:AN_cont_lemma_ilc_mpc_bound_1}.
with $\gamma_8,g_8\in\kinf$. Therefore, returning to \cref{eq:AN_cont_lemma_xk_m_r_bound}, we see that $\|x_k(t)-r(t)\|\leq h_4(\c)+\sigma_4(1/k)$ where $h_4(\cdot):=\bar{m}(t\!-\!1)\gamma_1(\cdot)+\gamma_8(\cdot)+\|B(t\!-\!1)K(t\!-\!1)\|[h_1(\cdot)+\gamma_7(\cdot)]\in\kinf$ and $\gamma_3(\cdot):=\bar{m}(t-1)s_1(\cdot)+g_8(\cdot)+\|B(t\!-\!1)K(t\!-\!1)\|[\sigma_1(\cdot)+g_7(\cdot)]$. We therefore proved that \cref{eq:AN_cont_lemma_result} holds at time $t$, reaching a contradiction.
\end{proof}
\begin{proof}[Proof (ii)]
Assume that for all $\epsilon_1>h_4(\c)+\sigma_4(1/k)$ there exists some $\delta>0$ such that $\|x_k(i-1)-r(i-1)\|\leq\epsilon_1$ if $\|w_k(t)\|\leq\delta_1$ for $i<i-1$. We know from \cite[Theorem 3]{klatte1985lipschitz} that the optimal set of a convex quadratic program is outer Lipschitz continuous for linear perturbations of the cost function and constraints. As a result, the values of $z_k^{i-1}(0)$, $z_k^{i-1}(1)$, and $v_k^{i-1}(0)$ are outer Lipschitz continuous under perturbations of $x_k(i-1)$, meaning that, if we denote with $z_k^{i-1,0}(0)$, $z_k^{i-1,0}(1)$, and $v_k^{i-1,0}(0)$ any optimizers of \Pmpc for any $x_k(i-1)$ satisfying $\|x_k(i-1)-r(i-1)\| \leq h_4(\c)+\sigma_4(1/k)$, and with $z_k^{i-1,\nu}(0)$, $z_k^{i-1,\nu}(1)$, and $v_k^{i-1,\nu}(0)$ any optimizers of \Pmpc for $x_k(i-1)=r(i-1)+\nu$, we have
\begin{align*}
\|z_k^{i-1,\nu}(0)-z_k^{i-1,0}(0)\| \leq k_1\|\nu\|,\\
\|z_k^{i-1,\nu}(1)-z_k^{i-1,0}(1)\| \leq k_2\|\nu\|,\\
\|v_k^{i-1,\nu}(0)-v_k^{i-1,0}(0)\| \leq k_3\|\nu\|,
\end{align*}
for some $k_1,k_2,k_3 \geq 0$. Combining with our previous results, we get that for all $\|x_k(i-1)-r(i-1)\| \leq \delta_2$ the following holds
\begin{align*}
\|z_k(i-1)-z_k^{i-1}(0)\|&\leq\gamma_7(\c)+g_7(1/k)+k_1\delta_2,\\
\|z_k(i)-z_k^{i-1}(1)\|&\leq\gamma_7(\c)+g_7(1/k)+k_2\delta_2,\\
\|B(i-1)[v_k(i-1)-v_k^{i-1}(0)]\|&\leq\gamma_8(\c)+g_8(1/k)\\&\phantom{\leq}~+(k_1+k_3)\delta_2,
\end{align*}
and combining with \cref{eq:AN_cont_lemma_xk_m_r_bound} we conclude that
\begin{multline*}
\|x_k(i)- r(i)\| \leq h_4(\c) + \sigma_4(1/k) \\ + (\|B(i\!-\!1)K(i\!-\!1)\| k_1 + k_1 + k_3) \delta_2.
\end{multline*}
We conclude that for $\|x_k(i-1)-r(i-1)\| \leq \delta$ and $\|x_k(0)-\bar{x}\| \leq \delta$, with $\delta:=\min\{\delta_1,\delta_2\}$, we have $\|x_k(i)- r(i)\| \leq \epsilon$, where $\epsilon \leq h_4(\c) + \sigma_4(1/k) + (\|B(i\!-\!1)K(i\!-\!1)\| k_1 + k_1 + k_3) \delta$ can be taken arbitrarily close to $h_4(\c)+\sigma_4(1/k)$ by making $\delta\to 0$. This completes the proof.
\end{proof}
\begin{remark}
The proof could be extended to the case $p=2$ as long as the solution of each \Pmpc problem possesses the outer Lipschitz continuity property. This is the case if, for example, \Pmpc satisfies the second order sufficient conditions of optimality and the strict constraint qualification \cite[Theorem 4.3]{zhang2017characterizations}.
\end{remark}

Next, we need the following technical result.

\begin{lemma}\label{lemma:AN_technical_lemma}
Suppose $\{ x_k \}_{k\in \N}$ is a sequence of independent random variables satisfying for all $k\in\N$ and any $\epsilon>0$
\begin{align*}
\operatorname*{Pr}\{ \|x_k\| \leq h(c_1) + \gamma(1/k) + \epsilon \}>0,
\end{align*}
where $h,\gamma\in\kinf$ and $c_1>0$ is a constant. Then for any $p\in(0,1)$ there exists some $g\in\kinf$ such that with probability at least $p$ we have for all $k\in\N$ that
\begin{align}\label{eq:AN_tech_lemma_statement}
\operatorname*{Pr} \left\{ \min_{0 \leq j \leq k} \|x_j\| \leq h(c_1) + g(1/k) \right\}\geq p.
\end{align}
\end{lemma}
\begin{proof}
We have that \cref{eq:AN_tech_lemma_statement} is equivalent to
\begin{align*}
p \leq & \operatorname*{Pr}\left\{ \min_{0 \leq j \leq k} \|x_j\| \leq h(c_1) + g(1/k) \right\} \\
= \, & 1 - (\operatorname*{Pr}\{ \|x_j\| > h(c_1) + g(1/k) \})^k \\
= \, & 1 - (1 - \operatorname*{Pr}\{ \|x_j\| \leq h(c_1) + g(1/k) \} )^k
\end{align*}
and this holds if and only if
\begin{align*}
\operatorname*{Pr}\{ \|x_j\| \leq h(c_1) + g(1/k) \} \geq 1 - \sqrt[k]{1-p}.
\end{align*}
Let $f(\epsilon):=\operatorname*{Pr}\{ \|x_k\| \leq h(c_1)+\gamma(1/k)+\epsilon \}$. Then $f$ is strictly increasing in the range $[f(0),1]$ and so is its inverse. Let $\rho\in\kinf$ be defined as
\begin{align*}
\rho(1/k) = 
\begin{cases}
f^{-1}(1-\sqrt[k]{1-p}) & \text{if } 1-\sqrt[k]{1-p} \in [f(0),1]\\
0 & \text{otherwise}.
\end{cases}
\end{align*}
Since $f^{-1}$ is strictly increasing on its domain $[f(0),1]$, with $f^{-1}(f(0))=0$, $\rho$ is monotonically increasing (although not strictly increasing). Choosing $g:=\gamma+\rho\in\kinf$ completes the proof.
\end{proof}

In the upcoming proofs, we repeatedly use a trajectory $\tilde{\mathbf{x}}_k$ as a candidate solution to \Pilc. Given the solution $\mathbf{x}_k$ of \Pilc and some index $j\in\ztm$, the trajectory $\tilde{\mathbf{x}}_k$ is defined as $\tilde{\mathbf{x}}_k=\mathbf{x}_k$ with the following modifications
\begin{align}
\begin{split}\label{eq:AN_set_to_singleton_candidate_traj}
\tilde{\alpha}_k(j)&=\argmin_{\alpha\in\{0,1\}}\{\alpha\|x_k(j)-r(j)\|+(1\!-\!\alpha)\|x_k^\text{ref}(j)-r(j)\|\},\\
\tilde{v}_k(j)&=\argmin_{v\in\Ru}\{\|v-v^*(j)\|^2_2:(v,z_k(j))\in\mathscr{Z}^{\psi(j)\eta_k(j)}\},\\
&\hspace{-0.77cm}\tilde{x}_k^\text{ref}(j)= \argmin_{x\in\{x_k^\text{ref}(j),x_k(j)\}}\|x-r(j)\|,\\
&\hspace{-0.8cm}\tilde{d}_k^\text{ref}(j)=B(j)[v_k(j)-\tilde{v}(j)]+d_k^\text{ref}(j).
\end{split}
\end{align}
and with $\tilde{\rho}_k$ and $\tilde{\xi}_k$ chosen as small as possible.

In our next result, we prove that, under the same assumptions as \cref{lemma:AN_continuity_lemma}, the estimation of the disturbance associated to the state reference $\tilde{x}^\text{ref}_k(j)$ becomes arbitrarily good. This is a crucial step towards proving that the algorithm converges to a steady state where $\xi$ and $\rho$ become arbitrarily close to $0$.

\begin{lemma}%
\label{lemma:AN_disturbance_set_to_singleton}
Let $i\in\Z_{(1,T]}$, and $j\in\Z_{[0,i]}$ be fixed. Assume that there exists some $h_1,h_2,h_3\in\kinf$ such that for any probability $p\in(0,1)$
\begin{align*}
\|z_k(t)-r(t)\| &\leq h_1(\c)+\sigma_1(1/k) ,~~ \forall t\leq i,\\
\|\xi_k(t)\| &\leq h_2(\c)+\sigma_2(1/k),~~ \forall t<i,\\
\|\rho_k(t)\| &\leq h_3(\c)+\sigma_3(1/k),~~ \forall t<j,
\end{align*}
with probability $p$ for some $\sigma_1,\sigma_2,\sigma_3\in\kinf$. Then there exists some $h_5\in\kinf$ such that for any probability $p\in(0,1)$
\begin{align*}
\mathcal{D}_{k|k}(\tilde{x}_k^\text{ref}(j),j)\subset\{f(\tilde{x}_k^\text{ref}(j),j)\}\oplus \mathcal{B}(h_5(\c)+\sigma_5(1/k)),
\end{align*}
with probability $p$, for some $\sigma_5\in\kinf$ and where $\tilde{x}_k^\text{ref}(j)$ is defined as in \cref{eq:AN_set_to_singleton_candidate_traj}.
\end{lemma}
\begin{proof}
Let $p\in(0,1)$ be fixed, and suppose that as $k\to\infty$ and $\c\to 0$ we have
\begin{align*}
\mathcal{D}_{k|k}(\tilde{x}_k^\text{ref}(j),j)\ominus\{\f(\tilde{x}_k^\text{ref}(j),j)\}\to\mathcal{B}(r),
\end{align*}
for some $r>0$.
Using \cref{lemma:AN_continuity_lemma}, for any $\epsilon>h_4(\c)+\sigma_4(1/k)$ we have that $\|x_k(j)-r(j)\|\leq\epsilon$ with nonzero probability $\bar{p}=pp'$, where $p'$ is the probability of $\|x_k(0)-\bar{x}\|\leq \delta$ and $\|w_k(t)\|\leq \delta$ for $t<i$, where $\delta$ is defined in \cref{lemma:AN_continuity_lemma}. Moreover, from the choice of variables in \cref{eq:AN_set_to_singleton_candidate_traj}, we get that similarly $\|\tilde{x}^\text{ref}_k(j)-r(j)\|\leq\epsilon$ with probability $\bar{p}$. Combining, we obtain that $\|x_k(j)-\tilde{x}_k^\text{ref}(j)\|\leq 2\epsilon$ with probability $\bar{p}$.

For any $k$ we therefore have with probability $\bar{p}$ that
\begin{multline*}
d_k(j)\!+\!w_k(j)\oplus\mathcal{B}(\bar{w})\subseteq\f(\tilde{x}_k^\text{ref}(j),j)\!+\!w_k(j)\\\oplus\mathcal{B}(\bar{w}\!+\!2m(j)\epsilon),
\end{multline*}
Consider now the probability (over the random variable $w_k(j)$)
\begin{align*}
&~\operatorname{Pr}\left\{ \f(\tilde{x}_k^\text{ref}(j),j) \oplus \mathcal{B}(r) \subseteq d_k(j)+w_k(j) + \mathcal{B}(m(j) \right. \\ & \left. \hspace{4.2cm} \cdot\, \|\tilde{x}_k^\text{ref}(j)-x_k(j)\|+\bar{w}) \right\} \\ 
\leq &~ \bar{p} \operatorname{Pr}\left\{ \{\f(\tilde{x}_k^\text{ref}(j),j)\} \subseteq \{\f(\tilde{x}_k^\text{ref}(j),j)\} + w_k(j) + \mathcal{B}(4m(j) \right. \\ & \left. \hspace{3.4cm} \cdot\, \epsilon + \bar{w} - r) \vphantom{\tilde{x}_k^\text{ref}(j)} \right\} \\
=&~ \bar{p} \operatorname{Pr}\left\{ -w_k(j) \in \mathcal{B}(4m(j)\epsilon+\bar{w}-r) \right\}\\
=&~ \bar{p} \operatorname{Pr}\left\{ \|w_k(j)\| \leq 4m(j)\epsilon+\bar{w}-r \right\}.
\end{align*}
Therefore,
\begin{align*}
&~\operatorname{Pr}\left\{ \f(\tilde{x}_k^\text{ref}(j),j) \oplus \mathcal{B}(r) \not\subset d_k(j)+w_k(j) +\mathcal{B}(m(j) \right. \\ & \hspace{4.2cm} \left. \cdot\, \|\tilde{x}_k^\text{ref}(j)-x_k(j)\|+\bar{w}) \right\} \\ 
\geq & ~ \bar{p} \operatorname{Pr}\left\{ \|w_k(j)\| \geq 4m(j)\epsilon+\bar{w}-r \right\},
\end{align*}
and this probability is strictly greater than $0$ whenever $r>4m(j)\epsilon$ thanks to \cref{ass:PF_disturbance}. This means for any value of $r>0$, we can find some value of $\c$ small enough and some value of $k$ large enough such that $\epsilon>h_4(\c)+\sigma_4(1/k)$ is sufficiently small and $r>4m(j)\epsilon$. This proves that for every $k$, given any $r>h_5(c_1)+g_9(1/k)$, where $h_5(\cdot):=4m(j)h_4(\cdot)$ and $\sigma_5(\cdot):=4m(j)\sigma_4(\cdot)$, there is a nonzero probability that the set $\mathcal{D}_{k|k}(\tilde{x}^\text{ref}(j),j)$ will be strictly contained in $\{\f(\tilde{x}_k^\text{ref}(j),j)\} \oplus \mathcal{B}(r)$. Using \cref{lemma:AN_technical_lemma} we conclude that for a given probability $p$, there exists some $\sigma_5\in\kinf$ such that for all $k$ with probability $p$
\begin{align*}
\mathcal{D}_{k|k}(\tilde{x}_k^\text{ref}(j),j)\subset\{f(\tilde{x}_k^\text{ref}(j),j)\}\oplus \mathcal{B}(h_5(\c)+\sigma_5(1/k)).
\end{align*}
\end{proof}

In our final technical Lemma, we show that if $z_k$ asymptotically matches $r(t)$ up to some time step $i$, then the state deviation $\xi_k(t)$ and the estimation error $\rho_k(t)$ converge asymptotically to a value arbitrarily close to $0$ for all time steps $t < i$. To prove this result, we use both our previous technical Lemmas.

\begin{lemma}
\label{lemma:AN_xi_and_rho_are_zero}
Let $i\in\Z_{(1,T]}$ be fixed. Assume that there exists some $h_1\in\kinf$ such that for any probability $p\in(0,1)$
\begin{align}
\|z_k(t)-r(t)\|_p\leq h_1(\c)+\sigma_1(1/k),~~ \forall t\in\Z_{[0,i]}, \label{eq:AN_xi_and_rho_hp}
\end{align}
with probability $p$ for some $\sigma_1\in\kinf$. Then there exist $h_2,h_3\in\kinf$ such that for any probability $p\in(0,1)$
\begin{align*}
\|\xi_k(t)\|_p &\leq h_2(\c)+\sigma_2(1/k),~~ \forall t<i,\\
\|\rho_k(t)\|_p &\leq h_3(\c)+\sigma_3(1/k),~~ \forall t<i,
\end{align*}
with probability $p$ for some $\sigma_2,\sigma_3\in\kinf$.
\end{lemma}
\begin{proof}
We use induction. Let $j\in\Z_{[0,i)}$ and suppose that for any $p\in(0,1)$ there exist $\sigma_2,\sigma_3\in\kinf$ such that $\|\xi_k(t)\|_p\leq h_2(\c)+\sigma_2(1/k)$ and $\|\rho_k(t)\|_p\leq h_3(\c)+\sigma_3(1/k)$ for $t < j$ with probability $p$. We now prove that the same holds for $t=j$.
Consider the trajectory $\tilde{\mathbf{x}}$ defined in \cref{eq:AN_set_to_singleton_candidate_traj}. 

It's easy to see that the choice of $\tilde{v}(j)$ and $\tilde{d}^\text{ref}(j)$ produce $\tilde{z}=z_k$. Next, using \cref{eq:AN_xi_and_rho_hp} and \cref{lemma:AN_continuity_lemma} we have that with probability $p$
\begin{align}
\|x_k(j)-z_k(j)\|_p&\leq\|x_k(j)-r(j)\|_p+\|z_k(j)-r(j)\|_p\notag\\&\leq(h_4+h_1)(\c)+(\sigma_4+\sigma_1)(1/k)\notag\\&=:\gamma_{14}(\c)+g_{14}(1/k),\label{eq:AN_xi_and_rho_are_zero_bound_z_minus_x}
\end{align}
and therefore, thanks to the choice of $\tilde{\alpha}(j)$ we get $\|x_k^\text{ref}(j)-z_k(j)\|_p\leq\gamma_{14}(\c)+g_{14}(1/k)$ with probability $p$. Next, since the assumptions of \cref{lemma:AN_disturbance_set_to_singleton} are met, and since $\tilde{x}^\text{ref}(j)$ as defined in the proof of \cref{lemma:AN_disturbance_set_to_singleton} coincides with the one defined in \cref{eq:AN_set_to_singleton_candidate_traj} through the choice of $\tilde{\alpha}(j)$, we have that $\tilde{\rho}(j)\leq h_5(\c)+\sigma_5(1/k)$, and combining
\begin{align*}
m(j)\tilde{\xi}_k(j)+\tilde{\rho}_k(j) \leq \gamma_{15}(\c)+g_{15}(1/k),
\end{align*}
with probability $p$, with $\gamma_{15}:=h_5+\gamma_{14}$ and $g_{15}:=\sigma_5+g_{14}$. Since we must have $\mathcal{J}_\text{ILC}(\mathbf{x}_k)\leq\mathcal{J}_\text{ILC}(\tilde{\mathbf{x}})$ or else $\tilde{\mathbf{x}}$ would be feasible (since $\tilde{\eta}_k(j)<\eta_k(j)$) and with a lower cost than $\mathbf{x}$, we conclude that with probability $p$
\begin{align*}
m(j)\xi_k(j)+\rho_k(j) \leq \gamma_{15}(\c)+g_{15}(1/k),
\end{align*}
which in particular means that $\|\xi_k(j)\|_p\leq h_2(\c)+\sigma_2(1/k)$ and $\|\rho_k(j)\|_p\leq h_3(\c)+\sigma_3(1/k)$ with $h_2:=\gamma_{15}/m(j)$, $h_3:=\gamma_{15}$, $\sigma_2:=g_{15}/m(j)$, $\sigma_3:=g_{15}$ with probability $p$. This concludes the induction step. The initialization step of the proof follows from the same argument by replacing $j$ with $1$, and it is therefore omitted. This completes the proof.
\end{proof}

Before proving the main theorem, we need the following intermediate result.

\begin{proposition}\label{prop:AN_thm_intermediate}
Suppose \cref{ass:PF_noise,ass:PF_disturbance,ass:PF_initial_condition,ass:AN_initial_trajectories,ass:AN_perfect_tracking,ass:AN_strict_feasibility} hold. Then there exist $h_{1,i},h_{2,i},h_{3,i}\in\kinf$, $i\in\ztm$, and $h_{1,T}\in\kinf$ such that for any probability $p\in(0,1)$
\begin{subequations}
\begin{align}
\|z_k(t)-r(t)\|_p &\leq h_{1,t}(\c)+\sigma_{1,t}(1/k), \label{eq:AN_convergence_hp_1} \\
\xi_k(t) &\leq h_{2,t}(\c)+\sigma_{2,t}(1/k), \label{eq:AN_convergence_hp_2} \\
\rho_k(t) &\leq h_{3,t}(\c)+\sigma_{3,t}(1/k), \label{eq:AN_convergence_hp_3}
\end{align}
\end{subequations}
with probability $p$ for some $\sigma_{1,i},\sigma_{2,i},\sigma_{3,i}\in\kinf$, $i\in\ztm$, and
\begin{align*}
\|z_k(T)-r(T)\|_p \leq h_{1,T}(\c)+\sigma_{1,T}(1/k).
\end{align*}
with probability $p$ for some $\sigma_{1,T}\in\kinf$.
\end{proposition}

\begin{proof}
All statements in this proof are meant to hold with probability $p$. We use induction. First, the initialization step follows immediately from \cref{ass:PF_initial_condition} as $z_k(0)=\x=r(0)$. Next, suppose that for every $p$ there exists some $\sigma_{1,t}\in\kinf$ such that $\|z_k(t)-r(t)\|_p\leq h_{1,t}(\c)+\sigma_{1,t}(1/k)$ for $t\in\Z_{[0,i]}$. We now prove that for the same $p$ there exists some $h_{1,i+1}\in\kinf$ independent of $p$, and some $\sigma_{1,i+1}\in\kinf$ such that $\|z_k(i+1)-r(i+1)\|_p\leq h_{1,i+1}(\c)+\sigma_{1,i+1}(1/k)$.

First, note that thanks to \cref{lemma:AN_xi_and_rho_are_zero}, there exist $h_{2,t},h_{3,t},\sigma_{2,t},\sigma_{3,t}$ for $t\in\Z_{[0,i-1]}$ such that $\xi_k(t)\leq h_{2,t}(\c)+\sigma_{2,t}(1/k)$ and $\rho_k(t)\leq h_{3,t}(\c)+\sigma_{3,t}(1/k)$ for all $k$ and $t\in\Z_{[0,i-1]}$. Next, consider the trajectory $\tilde{\mathbf{x}}_k$ defined as $\tilde{v}_k=v_k$, $\tilde{d}_k^\text{ref}=d_k^\text{ref}$, $\tilde{x}_k^\text{ref}=x_k^\text{ref}$, with the following modifications: $\tilde{v}_k(i)$, $\tilde{x}_k^\text{ref}(i)$, $\tilde{\alpha}_k(i)$ defined as in \cref{eq:AN_set_to_singleton_candidate_traj} with $j=i$,
\begin{align}
\begin{split}
\tilde{d}_k^\text{ref}(i)&=f(\tilde{x}_k^\text{ref}(i),i),\\
\tilde{v}_k(i+1)&=v_k(i+1)+K(i+1)[\tilde{z}_k(i+1)-z_k(i+1)],\\
\tilde{d}_k^\text{ref}(i+1)&=A_K(i+1)[z_k(i+1)-\tilde{z}_k(i+1)]+d_k^\text{ref}(i+1),
\end{split}\label{eq:AN_thm_candidate_traj}
\end{align}
$\tilde{z}_k$ determined uniquely from the choice of $\tilde{d}_k^\text{ref},\tilde{v}_k^\text{ref}$, and $\tilde{\xi}_k,\tilde{\rho}_k$ chosen as small as possible while maintaining feasibility of the problem (we will specify their exact value in \cref{eq:AN_xi_rho_precise} and \cref{eq:AN_xi_rho_precise_2}).

Since the assumptions of \cref{lemma:AN_disturbance_set_to_singleton} hold with $j=i$, we have
\begin{align*}
\mathcal{D}_{k|k}(\tilde{x}_k^\text{ref}(i),i)\subset\{f(\tilde{x}_k^\text{ref}(i),i)\}\oplus \mathcal{B}(h_5(\c)+\sigma_5(1/k)).
\end{align*}
Moreover, using the same argument as in \cref{eq:AN_xi_and_rho_are_zero_bound_z_minus_x}, we get that
\begin{subequations}
\begin{align}
\|x_k^\text{ref}(i)-r(i)\|&\leq\gamma_{15}(\c)+g_{15}(1/k), \label{eq:AN_thm_bound_1} \\
\|x_k(i)-r(i)\|&\leq\gamma_{15}(\c)+g_{15}(1/k), \label{eq:AN_thm_bound_2}
\end{align}
\end{subequations}
for some $\gamma_{15},g_{15}\in\kinf$. Next, consider that
\begin{align}
&\|\tilde{z}_k(i+1)-r(i+1)\|\notag\\
\stackrel{(a)}{=}\,&\|A(i)[z_k(i)-r(i)]+B(i)[\tilde{v}_k(i)-v^*(i)]+\tilde{d}_k^\text{ref}(i)-d^*(i)\|\notag\\
\stackrel{(b)}{\leq}\,&\|A(i)\|[h_{1,i}(\c)+\sigma_{1,i}(1/k)]+\|B(i)\|\|\tilde{v}_k(i)-v^*(i)\|\notag\\&+m(i)\|x_k^\text{ref}(i)-r(i)\|\notag\\
\stackrel{(c)}{\leq}\,&\|A(i)\|[h_{1,i}(\c)+\sigma_{1,i}(1/k)]+\|B(i)\|\|\tilde{v}_k(i)-v^*(i)\|\notag\\&+m(i)[\gamma_{15}(\c)+g_{15}(1/k)],\label{eq:AN_thm_bound_z_til_minus_r_1}
\end{align}
where $(a)$ follows from the definition of $\tilde{z}_k(i+1)$ and $r(i+1)$, $(b)$ follows from \cref{ass:PF_disturbance} and \cref{eq:AN_convergence_hp_1}, and in $(c)$ we used \cref{eq:AN_thm_bound_1}.
We now prove that there exists some $\gamma_{16},g_{16}\in\kinf$ such that $\|\tilde{v}_k(i)-v^*(i)\|\leq \gamma_{16}(\c)+g_{16}(1/k)$. To see this, observe that $\tilde{v}_k(i)$ is the unique solution to the strongly convex optimization problem
\begin{align*}
\operatorname*{minimize}_v& \quad \|v-v^*(i)\|^2\\
\text{subject to}& \quad H_u v \leq h-\psi(i)\tilde{\eta}_k(i)-H_x \tilde{z}_k(i),
\end{align*}
where the variables $\tilde{\eta}_k(i)$ and $\tilde{z}_k(i)$ are parameters of the problems entering linearly in the inequality constraints. As a result, because of the outer Lipschitz continuity of $\tilde{v}_k(i)$ with respect to joint variations of $\tilde{\eta}_k(i)$ and $\tilde{z}_k(i)$ (as per \cite[Theorem 3]{klatte1985lipschitz}), we get that the solution $\tilde{v}_k(i)$ satisfies
\begin{align*}
\|\tilde{v}_k(i)-v^*(i)\| \leq k_v [\|\tilde{\eta}_k(i)-\eta^*(i)\|+\|\tilde{z}_k(i)-r(i)\|],
\end{align*}
for some $k_v \geq 0$, where we have used the fact that $v^*(i)$ solves the problem with $\tilde{\eta}_k(i)=\eta^*(i)$ and $\tilde{z}_k(i)=r(i)$. Since $\tilde{\xi}_k(t)=\xi_k(t)\leq h_{2,t}(\c)+\sigma_{2,t}(1/k)$ and $\tilde{\rho}_k(t)=\rho_k(t)\leq h_{3,t}(\c)+\sigma_{3,t}(1/k)$ for $t\in\Z_{[0,i-1]}$, there exist $\gamma_{17},g_{17}\in\kinf$ such that $\tilde{\eta}_k(i)\leq\gamma_{17}(\c)+g_{17}(1/k)+\eta^*(i)$,
where $\eta^*(i)$ is defined in \cref{ass:AN_perfect_tracking}. We therefore get
\begin{align*}
&~\|\tilde{v}_k(i)-v^*(i)\|\\
\leq&~k_v[\gamma_{17}(\c)+g_{17}(1/k)+h_{1,i}(\c)+\sigma_{1,i}(1/k)]\\
=:&~\gamma_{16}(\c)+g_{16}(1/k).
\end{align*}
Replacing in \cref{eq:AN_thm_bound_z_til_minus_r_1} we obtain
\begin{align}\label{eq:AN_thm_bound_3}
\|\tilde{z}_k(i+1)-r(i+1)\|\leq \gamma_{18}(\c)+g_{18}(1/k),
\end{align}
where $\gamma_{18}(\cdot):=\|A(i)\|h_{1,i}(\cdot)+\|B(i)\|\gamma_{16}(\cdot)+m(i)\gamma_{15}(\cdot)$ and $g_{18}(\cdot):=\|A(i)\|\sigma_{1,i}(\cdot)+\|B(i)\|g_{16}(\cdot)+m(i)g_{15}(\cdot)$. Moreover, we have by definition
\begin{align*}
&~\tilde{z}_k(i+2)\\
=&~A(i+1)\tilde{z}_k(i+1)+B(i+1)\tilde{v}_k(i+1)+\tilde{d}_k^\text{ref}(i+1)\\
=&~A(i+1)\tilde{z}_k(i+1)+B(i+1)\tilde{v}_k(i+1)+B(i+1)K(i+1)\\&~\cdot[\tilde{z}_k(i+1)-z_k(i+1)]+A_K(i+1)[z_k(i+1)-\tilde{z}_k(i+1)]\\&~+d_k^\text{ref}(i+1),\\
=&~A(i+1)z_k(i+1)+B(i+1)v_k(i+1)+d_k^\text{ref}(i+1)\\
=&~z_k(i+2).
\end{align*}
The variables $\tilde{\xi}_k(i)$, $\tilde{\rho}_k(i)$, and $\tilde{\xi}_k(i+1)$, $\tilde{\rho}_k(i+1)$ associated to the choice in \cref{eq:AN_set_to_singleton_candidate_traj,eq:AN_thm_candidate_traj} can be chosen as
\begin{align}
\begin{split}
\tilde{\xi}_k(i) &= \gamma_{15}(\c) + g_{15}(1/k) ,\\
\tilde{\rho}_k(i) &= h_5(\c) + \sigma_5(1/k),
\end{split}\label{eq:AN_xi_rho_precise}
\end{align}
and
\begin{align}
\begin{split}
\tilde{\xi}_k(i+1) &= \|\tilde{z}_k(i+1)-\tilde{x}_k^\text{ref}(i+1)\|, \\
\tilde{\rho}_k(i+1) &= \|A_K(i+1)[z_k(i+1)-\tilde{z}_k(i+1)]\| + \rho_k(i+1),
\end{split}
\label{eq:AN_xi_rho_precise_2}
\end{align}
with $\tilde{\xi}_k(t)=\xi_k(t)$, $\tilde{\rho}_k(t)=\rho_k(t)$ for all $t \neq i, i+1$.

We now prove that the trajectory $\bar{\mathbf{x}}_k:=\theta \tilde{\mathbf{x}}_k+(1-\theta)\mathbf{x}_k$, for an appropriately small $\theta\in(0,1]$, produces a lower cost than $\mathbf{x}_k$. Note that thanks to \cref{ass:AN_strict_feasibility} there always exists some $\theta\in (0,1]$ for which $\bar{\mathbf{x}}_k$ is feasible. Since $\bar{\mathbf{x}}_k$ and $\mathbf{x}_k$ only differ at time-steps $i$ and $i+1$, we have
\begin{align*}
\mathcal{J}_\text{ILC}(\bar{\mathbf{x}}_k)-\mathcal{J}_\text{ILC}(\mathbf{x}_k) = \sum_{t\in\{i,i+1\}} \delta\mathcal{J}(t),
\end{align*}
where $\delta\mathcal{J}(t)$ is the difference in cost between $\mathcal{J}_\text{ILC}(\bar{\mathbf{x}}_k)$ and $\mathcal{J}_\text{ILC}(\mathbf{x}_k)$ at time step $t$.

We have
{\allowdisplaybreaks
\begin{align*}
&\delta\mathcal{J}(i)\\
=\,\,& c_2(i) \bar{\rho}_k(i) + \c \bar{\xi}_k(i) - c_2(i) \rho_k(i) - \c \xi_k(i) \\
\stackrel{(a)}{\leq} \,& \frac{\c}{m(i)} \|\theta d_k^\text{ref}(i)  + (1-\theta)\tilde{d}_k^\text{ref}(i) - f(\tilde{x}_k^\text{ref}(i),i)\| \\ & + \c [h_5(\c)+\sigma_5(1/k) + \gamma_{15}(\c)+g_{15}(\c)] \\ & - \frac{\c}{m(i)} \|d_k^\text{ref}(i) - f(x_k^\text{ref}(i),i)\| \\ & - \c\|z_k(i)- x_k^\text{ref}(i)\| \\
\stackrel{\cref{eq:AN_thm_candidate_traj}}{\leq} \, & \frac{\theta\c}{m(i)} \|d_k^\text{ref}(i)-\tilde{d}_k^\text{ref}(i)\| \\ & + \c [h_5(\c)+\sigma_5(1/k) + \gamma_{15}(\c)+g_{15}(\c)] \\ & - \frac{\c}{m(i)} \|d_k^\text{ref}(i) - f(x_k^\text{ref}(i),i)\| \\ & - \c\|z_k(i)- x_k^\text{ref}(i)\| \\
\stackrel{(b)}{\leq} \, & \frac{\theta\c}{m(i)} \|d_k^\text{ref}(i)-f(x_k^\text{ref}(i),i)\| \\ & + \frac{\theta\c}{m(i)} \|\tilde{d}_k^\text{ref}(i)-f(x_k^\text{ref}(i),i)\|\\ & + \c [h_5(\c)+\sigma_5(1/k) + \gamma_{15}(\c)+g_{15}(\c)] \\ & - \frac{\c}{m(i)} \|d_k^\text{ref}(i) - f(x_k^\text{ref}(i),i)\| \\ & - \c\|z_k(i)- x_k^\text{ref}(i)\| \\
\stackrel{\hphantom{(a)}}{\leq} \, & \frac{(\theta-1)\c}{m(i)} \|d_k^\text{ref}(i)-f(x_k^\text{ref}(i),i)\| \\ & +\theta\c \|\tilde{x}_k^\text{ref}(i)-x_k^\text{ref}(i)\| - \c\|z_k(i)- x_k^\text{ref}(i)\| \\ & + \c [h_5(\c)+\sigma_5(1/k) + \gamma_{15}(\c)+g_{15}(\c)] \\
\stackrel{(c)}{\leq} \, & \frac{(\theta-1)\c}{m(i)} \|d_k^\text{ref}(i)-f(x_k^\text{ref}(i),i)\| \\ &  +(\theta-1)\c\|z_k(i)- x_k^\text{ref}(i)\| \\ & + \c [h_5(\c)+\sigma_5(1/k) + \gamma_{15}(\c)+g_{15}(\c)] \\ & + \theta\c[\gamma_{15}(\c)+g_{15}(1/k)+h_{1,i}(\c)+\sigma_{1,i}(1/k)],
\end{align*}}
where in $(a)$ we used that
\begin{align*}
\bar{d}_k^\text{ref}(i) &= \theta d_k^\text{ref}(i) + (1-\theta) \tilde{d}_k^\text{ref}(i),\\
\rho_k(i) & \geq \|d_k^\text{ref}(i)-f(x_k^\text{ref}(i),i)\|,
\end{align*}
and that
\begin{align*}
\bar{\rho}_k(i) = \, & \sup_{d\in \mathcal{D}(x_k^\text{ref}(i))} \|\theta d_k^\text{ref}(i)  + (1-\theta)\tilde{d}_k^\text{ref}(i)-d\| \\
\leq \, & \|\theta d_k^\text{ref}(i)  + (1-\theta)\tilde{d}_k^\text{ref}(i)-f(\tilde{x}_k^\text{ref}(i),i)\| \\ & + \sup_{d\in\mathcal{D}(x_k^\text{ref}(i))} \|f(\tilde{x}_k^\text{ref}(i),i)-d\|\\
\leq \, & \|\theta d_k^\text{ref}(i) + (1-\theta)\tilde{d}_k^\text{ref}(i)-f(\tilde{x}_k^\text{ref}(i),i)\| \\ & + h_5(\c) + \sigma_5(1/k),
\end{align*}
in $(b)$ we added and subtracted equal terms, and in $(c)$ we used that
\begin{align*}
\|\tilde{x}_k^\text{ref}(i)-z_k(i)\| &\leq \|\tilde{x}_k^\text{ref}(i)-r(i)\| + \|z_k(i)-r(i)\| \\
& \leq \gamma_{15}(\c)+g_{15}(1/k)+h_{1,i}(\c)+\sigma_{1,i}(1/k),
\end{align*}
where the last inequality follows from \cref{eq:AN_convergence_hp_1} and \cref{eq:AN_thm_bound_1}. We conclude that
\begin{align*}
\delta \mathcal{J}(i) \leq \, & \c [ \frac{\theta-1}{m(i)} \|d_k^\text{ref}(i)-f(x_k^\text{ref}(i),i)\|  \\ & + (\theta-1) \|z_k(i)- x_k^\text{ref}(i)\| ] \\ & + \c [\gamma_{19}(\c)+g_{19}(1/k)], \\ 
\leq \, & \c [\gamma_{19}(\c)+g_{19}(1/k)],
\end{align*}
where $\gamma_{19}:=h_5+2\gamma_{15}+h_{1,i}\in\kinf$, and $g_{15}:=\sigma_5+2g_{15}+\sigma_{1,i}\in\kinf$. Next, consider that by definition
\begin{align}
&\bar{\xi}_k(i+1)-\xi_k(i+1) \notag \\=\,&\|\theta z_k(i+1)+(1-\theta)\tilde{z}_k(i+1)-x_k^\text{ref}(i+1)\| \notag \\ & - \|z_k(i+1) -x_k^\text{ref}(i+1)\| \notag \\
\leq \, & (1-\theta) \|\tilde{z}_k(i+1)-z_k(i+1)\| \notag \\
\leq \, & (1-\theta) \|\tilde{z}_k(i+1)-z_k(i+1)\|_2, \label{eq:AN_thm_cost_1}
\end{align}
where the last step is satisfied with equality if $p=2$. Similarly, the estimation error satisfies by definition
\begin{align}
& \bar{\rho}_k(i+1)-\rho_k(i+1) \notag \\ \leq \, & (1-\theta) \|A_K(i+1)\| \|z_k(i+1)-\tilde{z}_k(i+1)\| \notag \\
\leq \, & (1-\theta) \|A_K(i+1)\|_2 \|z_k(i+1)-\tilde{z}_k(i+1)\|_2. \label{eq:AN_thm_cost_2}
\end{align}
Finally, the state error satisfies
\begin{align}
& \|\theta z_k(i+1) + (1-\theta)\tilde{z}_k(i+1)-r(i+1)\|_Q^2 \notag \\ &- \|z_k(i+1)-r(i+1)\|_Q^2 \notag \\
\stackrel{(a)}{=} & \, (1-\theta) \|\tilde{z}_k(i+1)-r(i+1)\|_Q^2 + (\theta-1)\|z_k(i+1) \notag \\ & -r(i+1)\|_Q^2-\theta(1-\theta)\|z_k(i+1)-\tilde{z}_k(i+1)\|_Q^2 \notag \\
\stackrel{(b)}{\leq} \, & 2(1-\theta)\lambda_\text{max}(Q)^2 \operatorname{diam}\mathscr{X}_x[\gamma_{18}(\c)+g_{18}(1/k)] \notag \\ & +(\theta-1)\|z_k(i+1)-r(i+1)\|_Q^2 \notag \\
\stackrel{(c)}{=:} \, & (1-\theta)[\gamma_{20}(\c)+g_{20}(1/k)] + (\theta-1)\|z_k(i+1) \notag \\ &-r(i+1)\|_Q^2, \label{eq:AN_thm_cost_3}
\end{align}
where in $(a)$ we used the identity $\|(1-\theta)a+\theta b\|_2^2=(1-\theta)\|a\|_2^2+\theta \|b\|_2^2-\theta(1-\theta)\|a-b\|_2^2$, in $(b)$ we used that $\|\tilde{z}_k(i+1)-r(i+1)\|_2 \leq 2\operatorname*{diam}\mathscr{X}_x$, and \cref{eq:AN_thm_bound_3}, and in $(c)$ we defined $\gamma_{20}:=2\lambda_\text{max}(Q)^2 \operatorname{diam}\mathscr{X}_x \gamma_{18}\in\kinf$ and $g_{20}:=2\lambda_\text{max}(Q)^2 \operatorname{diam}\mathscr{X}_x g_{18}\in\kinf$. Combining \cref{eq:AN_thm_cost_1}, \cref{eq:AN_thm_cost_2}, and \cref{eq:AN_thm_cost_3}, we get
{\allowdisplaybreaks
\begin{align*}
& \delta \mathcal{J}(i+1) \\
\stackrel{\hphantom{(a)}}{=} \, & c_1 [\bar{\xi}_k(i+1)-\xi_k(i+1)] + c_2(i+1) [\bar{\rho}_k(i+1)-\rho_k(i+1)] \\ & + \|\bar{z}_k(i+1)-r(i+1)\|_Q^2 - \|z_k(i+1)-r(i+1)\|_Q^2 \\
\stackrel{(a)}{\leq} \, & (1-\theta)\c \|\tilde{z}_k(i+1)-z_k(i+1)\|_2 \\ & + \frac{(1-\theta)\c}{m(i+1)} \|A_K(i+1)\|_2 \|z_k(i+1)-\tilde{z}_k(i+1)\|_2 \\ & + (1-\theta)[\gamma_{20}(\c)+g_{20}(1/k)] \\ & + (\theta-1)\lambda_\text{min}(Q)^2\|z_k(i+1)-r(i+1)\|_2^2 \\
\stackrel{\hphantom{(b)}}{=} \, & \|\tilde{z}_k(i+1)-z_k(i+1)\|_2(1-\theta)[\c+\frac{\c\|A_K(i+1)\|_2}{m(i+1)}\\&-\lambda_\text{min}(Q)^2\|z_k(i+1)-r(i+1)\|_2] \\& + (1-\theta)[\gamma_{20}(\c)+g_{20}(1/k)] \\
\stackrel{(b)}{\leq} \, & (1-\theta) \bigg[2\operatorname{diam}\mathscr{X}_x [\c+\frac{\c\|A_K(i+1)\|_2}{m(i+1)}\\&-\lambda_\text{min}(Q)^2\|z_k(i+1)-r(i+1)\|_2] \\& + \gamma_{20}(\c)+g_{20}(1/k)\bigg],
\end{align*}}
where $(a)$ follows from \cref{eq:AN_thm_cost_3} and from the definitions of $\mathbf{x}_k$ and $\bar{\mathbf{x}}_k$, and $(b)$ follows from $\|\tilde{z}_k(i+1)-r(i+1)\|_2 \leq 2\operatorname*{diam}\mathscr{X}_x$. Note that if the following condition holds
\begin{align*}
\|z_k(i+1)-r(i+1)\|_2 > \, & \frac{c_1 \|A_K(i+1)\|_2}{\lambda_\text{min}(Q)^2m(i+1)} \\ & + \frac{\gamma_{20}(c_1)+g_{20}(1/k)}{2 \operatorname*{diam}\mathscr{X}_x\lambda_\text{min}(Q)^2}\\
=: \, & \gamma_{21}(\c) + g_{21}(1/k),
\end{align*}
holds, then for any $\theta\in(0,1]$ we have $\mathcal{J}_\text{ILC}(\bar{\mathbf{x}}_k)-\mathcal{J}_\text{ILC}(\mathbf{x}_k)<0$. This is of course not possible since it would violate the optimality of $\mathbf{x}_k$. Hence, we must have
\begin{align*}
\|z_k(i+1)-r(i+1)\|_2 \leq \gamma_{21}(\c) + g_{21}(1/k).
\end{align*}
Choosing $h_{1,i+1}:=k_{p,+}\gamma_{21}$ and $\sigma_{1,i+1}:=k_{p,+}g_{21}(1/k)$ completes the proof.
\end{proof}

We are now ready to prove the main result.

\begin{proof}[Proof of \cref{thm:AN_convergence}]
For any $t$, we have for any probability $p\in(0,1)$ that $\|z_k(t)-r(t)\| \leq h(c_1) + \sigma(1/k)$ for some $h\in\kinf$ independent of $p$ and for some $\sigma\in\kinf$. Therefore, for any $p\in(0,1)$ we have $\limsup_{k\to\infty} \|z_k(t)-r(t)\| \leq h(c_1)$. Since this statement holds for any $p\in(0,1)$, it holds with probability $1$.
\end{proof} 

\end{document}